\documentclass[12pt]{amsart}

\usepackage[utf8]{inputenc}
\usepackage[T2A]{fontenc}
\usepackage[english]{babel}
\usepackage{amsfonts,amsmath, amssymb}

\usepackage{csquotes}

\usepackage{xcolor}

\usepackage{enumerate}
\usepackage{epsf,epsfig,amsfonts,a4wide}
\usepackage{amsfonts, mathdots}
\usepackage{amsmath,amssymb,amsthm}
\usepackage{url}
\usepackage{amsmath,amssymb,amsthm}

\usepackage{amssymb}
\usepackage{amsthm}
\usepackage{epsf,epsfig,amsfonts,graphicx}

\usepackage{hyperref}

\newcommand{\weg}[1]{}

\newtheorem{Theorem}{Theorem}[section]
\newtheorem{Fact}{Fact}[section]

\newtheorem{Lemma}{Lemma}[section]
\newtheorem{Proposition}{Proposition}[section]
\newtheorem{Corollary}{Corollary}[section]

\newtheorem{Ex}{Example}[section]
\newtheorem{Remark}{Remark}[section]

\theoremstyle{remark}

\newcommand{\be}{\begin{equation}}
\newcommand{\ee}{\end{equation}}

\newcommand{\R}{\mathbb{R}}\newcommand{\Id}{\textrm{\rm Id}}

\newcommand{\gl}{\mathrm{gl}}

\newcommand{\ddd}{\mathrm{d}}

\newcommand{\pd}[2]{\frac{\partial#1}{\partial#2}}

\newcommand{\dd}{{\mathrm d}\,}

\newcommand{\tr}{\operatorname{tr}}

\newcommand{\trace}{\operatorname{tr}}

\title[Finite-dimensional reductions of multicomponent integrable PDEs]{Finite-dimensional reductions and finite-gap type solutions of multicomponent integrable PDEs}
\author{Alexey V. Bolsinov}\address{ School of Mathematics,
 Loughborough University, Loughborough
 LE11 3TU, UK   and  Institute of Mathematics and Mathematical Modeling, Almaty, Kazakhstan}\email{\tt A.Bolsinov@lboro.ac.uk}  

\author{Andrey Yu. Konyaev}\address{Faculty of Mechanics and Mathematics, Moscow State University, 119992, Moscow, Russia}
 \email{\tt  maodzund@yandex.ru}
 \author{Vladimir S. Matveev}\address{
Institut f\"ur Mathematik, Friedrich Schiller Universit\"at Jena,
07737 Jena, Germany} \email{\tt  vladimir.matveev@uni-jena.de}

\usepackage{biblatex} %Imports biblatex package
\begin{document}

\begin{abstract}
The main object of the paper is a recently discovered family of  multicomponent integrable systems of partial differential equations, whose particular cases include many well-known equations such as the Korteweg--de\,Vries, coupled KdV, Harry Dym, coupled Harry Dym, Camassa--Holm, multicomponent Camassa--Holm, Dullin--Gottwald--Holm, and Kaup--Boussinesq equations.
    
We suggest a methodology for constructing a series of solutions for all systems in the family. 
The crux of the approach lies in reducing this system to a dispersionless integrable system
which is a special case of linearly degenerate quasilinear systems actively explored since the 1990s and recently studied  in the framework of Nijenhuis geometry.
These infinite-dimensional  integrable   systems are closely connected  to certain   explicit finite-dimensional integrable systems. We provide a link between solutions of our multicomponent PDE systems and solutions of this finite-dimensional  system, and use it  to   construct animations of multi-component analogous of soliton and cnoidal solutions.

\end{abstract}
\maketitle
\tableofcontents

\section{Introduction  }
\subsection{What systems are we studying?} \label{sec:1.1}

We deal with a family of  $n$-component  integrable  systems of PDEs  
constructed by the authors in \cite{nijapp4} within the Nijenhuis geometry project \cite{nij1}.  One of the main features of this family is that  it includes, as particular cases with appropriately chosen parameters, many well-known equations such as KdV, coupled KdV, Harry Dym, coupled Harry Dym, Camassa--Holm, multicomponent Camassa--Holm, Dullin--Gottwald--Holm, and Kaup--Boussinesq equations, but also contains many new systems (see \cite{nijapp4} for details).  
  
The family 
is constructed by  the following data:

\begin{itemize} 
\item Number of components, $n\in \mathbb N$. 
  \item Differentially-nondegenerate Nijenhuis operator\footnote{{\it Nijenhuis operator} is a (1,1)-tensor field $L$ whose Nijenhuis torsion $\mathcal{N}_L$ vanishes. A Nijenhuis operator is differentially-nondegenerate \cite[\S 4.2]{nij1}, if 
    the  coefficients of its characteristic polynomial  are functionally independent.}  $L$   in dimension $n$.
 
  \item  Polynomial 
  $m(\mu)$ of degree    $\le n$. 
\end{itemize}

Systems in the family are parametrised by a real (optionally, complex) parameter $\lambda$ and depending on this parameter, can be of 4 types which, for brevity,  we refer to as BKM I -- IV. 

The theoretical results of the present paper   can be 
equally applied to all four types. In the introduction, we will  concentrate on BKM IV systems only, as special cases of BKM IV systems  were better  studied previously  so it easier to find parallels between our  results and some classical results in the integrable systems  theory.  We start by  recalling   the construction of  BKM IV systems.

Consider $\mathbb{R}^{n}(u_1,...,u_n)$.
Choose a   polynomial
$m(\mu)= m_{n-1}\mu^{n-1} + m_{n-2}\mu^{n-2} \cdots + m_{0}$ of degree $\le n-1$    and  an explicitly given     differentially nondegenerate Nijenhius operator $L^i_j= L(u)$.    Set 
\begin{equation} \label{eq:condabove}
 \quad  \sigma(\mu)=\det (L-\mu \Id),  \quad \mathcal L_{\zeta_0} \sigma=1, \quad   \zeta = m(L) \zeta_0.
\end{equation}
As $L$ is differentially nondegenerate, the coefficients of $\sigma$ are  independent and conditions \eqref{eq:condabove}
uniquely determine  the vector field $\zeta$. The explicit form of $\zeta$ in coordinates will be given in examples below.

Next, consider the following system of $n$ PDEs: 
\begin{equation}
\label{eq:explform}
u_{t} =  \tfrac{1}{2} \left(\trace L \right)_{xxx}   \zeta + 
\left(L + \tfrac{1}{2} \trace L \cdot \Id\right) u_x.   
\end{equation}
Since  $L$ is explicitly  given in terms of $u$ and   $\trace L$  and the components of $\zeta$ are explicit functions of $u$,  then   \eqref{eq:explform} can be  rewritten in the ``dynamical system'' form 
\begin{equation}
\frac{\partial u_i(t,x)}{\partial t}  = V_i[u], \label{eq:dynform} \end{equation}
where the components of $V_i[u]$  are  explicit polynomials  in $ \tfrac{\partial u_i}{\partial x},\dots, \tfrac{\partial^3 u_i}{\partial x^3}$ with  coefficients  depending  on $u$. 

 This is a Kovalevskaya-type  system with  initial condition defined by a curve $x\mapsto u(x,0)$.  In the real-analytic case, Cauchy-Kovalevskaya Theorem guaranties the existence of  local solutions $u(t,x)$. We may view each solution $u(t,x) $ 
 as a family of curves $x\mapsto u(x,t)$ with $t$ being a parameter of the family. From this viewpoint,   the equation describes the evolution of a curve with time $t$.

\begin{Ex}[KdV as BKM IV]{\rm
    Take  $n = 1$ and differentially non-degenerate 1-dimensional Nijenhuis operator   $L = (u)$; the polynomial  $m$ has   degree $\le 0$ so that $m(\mu) = m_0\in \mathbb{R}$. Then, $\sigma(\mu)= u - \mu$,  $\zeta_0= \frac{\partial}{\partial u}$ and   $\zeta = m_0 \frac{\partial}{\partial u}$. Indeed,  $\mathcal L_{\zeta_0} \sigma(\mu) = \tfrac{\partial }{\partial u}(u-\mu) =  1$ and  $\zeta=m(L)\zeta_0 = m_0 \Id (\frac{\partial}{\partial u}) =  m_0 \frac{\partial}{\partial u}$. 
 
The equation \eqref{eq:explform} reads then 
\begin{eqnarray}
 u_{t} 
&    =&   \tfrac{m_0}{2} u_{xxx}+ \tfrac{3}{2} uu_x ,    \label{eq:KdV}
\end{eqnarray}
which is one of equivalent versions of the famous KdV equation.
}\end{Ex} 

 We see that
 even  for $n=1$, the construction gives something interesting.  Let us recall another physically interesting example with   $n=2$. 

\begin{Ex}[Kaup-Boussinesq as BKM IV]{\rm \label{ex:KB}

Take  $n = 2$ and differentially non-degenerate Nijenhuis operator   
$$
L =   \begin{pmatrix} -u_1 &1 \\ -u_2 & 0\end{pmatrix}.
$$

As the polynomial  $m$ of degree $\le 1$ we again take  the constant polynomial $m(\mu)=m_0\in \mathbb{R}$ of degree $0$. We have $\sigma(\mu)= \mu^2 + u_1\mu + u_2$ so that 
$$
\zeta_0=\begin{pmatrix}0\\1\end{pmatrix}   \ \textrm{ and }   \ \ \ \zeta = \begin{pmatrix}0\\
m_0\end{pmatrix}.
$$

Equation \eqref{eq:explform} is then the system of two PDEs
\begin{equation} \label{eq:KB}
\begin{aligned}
(u_1)_t & = (u_2)_x - \frac{3}{2} u_1 (u_1)_x, \\
(u_2)_t & = \frac{m_0}{2} (u_1)_{xxx} - u_2 (u_1)_x - \frac{1}{2} u_1 (u_2)_x.
\end{aligned}
\end{equation}
 This is the famous  Kaup-Boussinesq  system,  see e.g. \cite[eqn. (4)]{Pavlov2014} or \cite[eqn. (2.32)]{MY79}. This system is also known as  dispersive water wave system, see e.g. \cite{fordy}. 
}\end{Ex}

Definition and construction of BKM I, II and III systems can be found in \cite[\S 2.1]{nijapp4} and will be recalled below.    BKM II system  also has the ``dynamical system'' form \eqref{eq:dynform}. BKM I and BKM III systems are  the so called {\it evolutionary systems with constrains}, see \cite[eqn. (2)]{nijapp4}.

\begin{Remark} \label{rem:1}{\rm
    BKM systems constructed in \cite{nijapp4} are more general than those considered in the present paper.  Informally speaking, the systems in \cite{nijapp4}  are constructed from block-diagonal Nijenhuis operators $L=L_0\oplus L_1\oplus \dots\oplus L_N$ where each block $L_i$, $i=1,\dots, N$ contributes to the construction with a certain natural weight $\ell_i$. In the present paper, we restrict ourselves to  the  case $L=L_0$, i.e., $N=0$.
}\end{Remark}

\subsection{ Integrability of BKM systems } \label{sec:1.2}

In the finite-dimensional case, there is a well-established notion of (Arnold-Liouville) integrability of a Hamiltonian system, see e.g. \cite{BMMT}:  A  Hamiltonian system on a symplectic manifold $M^{2N}$ is {\it  integrable, } if it possesses $N$  functionally 
independent integrals in involution. 

This  definition does not have much sense in the PDE setting, and different authors declare  different properties of a PDE system responsible for its integrability, see e.g. \cite{HSW}. 
Integrability of BKM systems was discussed in  \cite{nijapp4}.  These systems are multi-Hamiltonian, see 
\cite{nijapp2} for the discussion on the corresponding multi-dimensional pencil of compatible Poisson structures.   Conservation laws and symmetries for BKM systems, which are infinite-dimensional analogs of first integrals and commuting vector fields, 
were constructed in \cite[\S 2.2]{nijapp4}.

By some experts, the integrability of a PDE system is understood as a possibility to find infinitely many qualitatively different solutions, see e.g.  \cite{Deift2019}. The goal of  the present paper is to construct such solutions for BKM systems. The construction goes through an explicit  reduction to certain  finite-dimensional integrable systems.

\subsection{Finite-dimensional reduction of BKM systems}  \label{sec:1.3}

By  {\it a finite-dimensional reduction} 
of an integrable PDE system, we understand an explicit  embedding of    a finite-dimensional  system,  integrable in the sense of  \S \ref{sec:1.2},  into our PDE system.  As the embedding is explicit,   
every solution of the  finite-dimensional system  gives a solution of the  infinite-dimensional system.

  For every $N$, we construct such an  
embedding of a certain integrable Hamiltonian system with $N$ degrees of freedom. 
The  Hamiltonian of this  finite dimensional system is   the  sum of potential and kinetic energies, 
$H= K_g + V$, 
 with explicit  flat metric $g $  and explicit  potential energy $V$. 
 The commuting integrals $F_0= - H, F_1, \cdots, F_{N-1}$ are also sums of   kinetic (e.g., quadratic in momenta)  and potential terms.

Let us now describe our finite-dimensional integrable system together with its embedding into BKM IV system.   Fix a BKM IV system, that is,  choose  $n$, $L$ and $m(\mu)$, as described in \S \ref{sec:1.1}.  Next,  
 choose $N\in \mathbb{N}$.  In the KdV case,  $N$ corresponds to the number of ``gaps'' in generalised  finite-gap  solutions of  KdV.

Next, consider the contravariant 
metric $g^{-1}=\Bigl( g^{ij}\Bigr) $   and operator $M$ on   $\mathbb{R}^N(w_1,\dots, w_N)$ :
\begin{equation} \label{eq:gM}
    g^{-1}= 
\begin{pmatrix}    
0 & \cdots & 0 & 0 & \!\!\!\! 1\\ 
0 &  \cdots & 0  &\! 1 & w_1\\ 
\vdots & \iddots &\iddots & \iddots  &w_{{2}} \\ 
0&\!\!1&w_{{1}}& \iddots &   \vdots    \\
1&w_{{1}}&w_{{2}}  &\cdots & \!\!\!  w_{{N-1}}
  \end{pmatrix}   ,  \quad  M = \begin{pmatrix}
-w_1 & 1 & 0 & \cdots  & 0 \\
-w_2 & 0 & 1 & \ddots &   \vdots\\
\vdots & \vdots & \ddots & \ddots & 0 \\
-w_{N-1} &  0 & \dots & 0 & 1\\
-w_N & 0  & \dots & 0 & 0 
\end{pmatrix}. 
\end{equation}

The metric  $g$ and operator $M$ enjoy the following useful properties  (see \cite{nijapp2}):  $g$ is flat and geodesically compatible with $M$ in the sense of  \cite[\S 6.2]{nij1} or \cite[\S 1.1]{nijapp5}. The latter property allows us to construct Poisson commuting integrals for the geodesic flow of $g$ using the formula from \cite{MT97,BM2003}: 

\begin{Fact} \label{fact:1}  
The functions $I_{0},\dots ,I_{N-1}:T^*\mathbb{R}^{N}\to \mathbb{R}$  defined by the polynomial relation 
\begin{equation}\label{eq:formulaintegrals}
 \frac{1}{2}   g^{-1}\left(\det(\mu\,\Id - M ) (M^* - \mu\,\Id)^{-1} p, p\right)= I_{0}(w,p)\mu^{N-1} {+} I_{1}(w,p)\mu^{N-2}{+} \cdots {+} I_{N-1}(w,p)  
\end{equation}
commute with respect to the canonical Poisson structure on $T^*\mathbb{R}^{N} (w_1,\dots, w_N, p_1,\dots,p_N).$
\end{Fact} 
Clearly, $I_0= - H_g = -\frac{1}{2} g^{ij}p_ip_j,$ so that $I_1,\dots, I_{n-1}$ 
are pairwise commuting integrals of the geodesic flow of $g$. 
Moreover, as $M$ is gl-regular\footnote{This means that the geometric multiplicity of each eigenvalue of $M$ equals one.}, the integrals $I_0,\dots,I_{N-1}$ are functionally independent almost everywhere by \cite[Lemma 5.6]{gover}.

Let us now construct the functions $U_0,...,U_{N-1}$ of the variables $w_1,\dots, w_{N}$ in such a way that the functions 
$F_i= I_i -  U_i$ still Poisson commute. We first explain the construction starting with an arbitrary real analytic function $f$. For reduced BKM systems,  functions $f$  will be specified explicitly in  \eqref{eq:Umu}. 

  Let $g$ and $M$ be geodesically compatible and $M$ be gl-regular. Choose a real analytic function $f$ of one variable and consider  the functions 
$U_0,U_2,\dots, U_{N-1}$ defined by the following matrix relation: 
\begin{equation} 
\label{eq:U}
f(M)=U_0M^{N-1}+U_1M^{N-2}+\dots+U_{N-1}\Id.
\end{equation}
Note that the left  hand side is an analytic function of $M$ which is a well-defined $(1,1)$-tensor field,  see e.g. \cite[\S 3.1]{nij1} for a discussion on real  analytic functions of Nijenhuis operators and in particular for the conditions under which they are well-defined. 
\weg{
The functions $f$ coming from finite-dimensional reductions of BKM systems will satisfy these conditions.}  
Relation \eqref{eq:U} can be understood as a  system of linear equations on functions $U_0,...,U_{N-1} $.  As $M$ is  gl-regular,  this system admits a unique solution, so that $U_0,\dots, U_{n-1}$ are uniquely defined from $f$.    

We need the following observation which can be verified by direct calculations. 

\begin{Fact} \label{fact:2}
For almost every point,   $M$ has $N$ 
different eigenvalues  which we call $q_1,\dots, q_N$. As $M$ is differentially nondegenerate, $q_i$  are 
functionally independent and give a local coordinate system. In this coordinate system, the metric $g=\Bigl( g_{ij}\Bigr)$  and operator $M$ are as  follows: 
\begin{equation}\label{eq:g_and_M}
g = \sum_{i=1}^N \left(\prod_{s=1, s\ne i}^N (q_i-q_s)  \right) \dd q_i^2  \ , \  \  M = \operatorname{ diag}(q_1,q_2,\dots, q_N) .   
\end{equation}
Moreover, the functions $F_i= I_i+ U_i$ $(i=0,\dots, N-1)$  satisfy the following ``St\"ackel'' relation: 
\begin{equation}\label{eq:stackel}
\begin{pmatrix} q_1^{N-1} & q_1^{N-2}& \cdots & 1 \\ 
                               q_2^{N-1} & q_2^{N-2}& \cdots & 1 \\
                                   \vdots & & & \vdots \\ 
                                   q_N^{N-1} & q_N^{N-2}& \cdots & 1\end{pmatrix} \begin{pmatrix} F_0 \\ 
                               F_1 \\
                                   \vdots   \\ 
                                  F_{N-1}\end{pmatrix} = 
\begin{pmatrix}  -\frac{1}{2}p_1^2 + f(q_1)\\ 
                                -\frac{1}{2}p_2^2 + f(q_2)\\
                                   \vdots   \\ 
                                  -\frac{1}{2}p_N^2+ f(q_N)\end{pmatrix},
\end{equation}
where $p_i$ are the momenta corresponding to the coordinates $q_i$. 
In particular,     $F_0, \dots, F_{N-1}$ Poisson commute and are functionally independent almost everywhere.  
\end{Fact}

\begin{Remark}{\rm
   The integrable systems constructed by \eqref{eq:stackel} are sometimes called {\it Benenti } or
   {\it Benenti-St\"ackel }
   systems, see e.g.  \cite{BM2006}. This is a well-studied class of finite-dimensional integrable systems. Finite dimensional reductions of various integrable PDEs are related to Benenti systems, see e.g. \cite{Dubrovin1975,VBM}.}   
\end{Remark}

For finite-dimensional reductions of BKM  IV systems, the function $f(\mu)$  is given by 
\begin{equation} \label{eq:Umu}
f(\mu) =  \frac{\mu^{2N+n} + c_{2}\mu^{2N+n-2}+ c_{3}\mu^{2N+n-3}+ \cdots + c_{2N+n}}{m(\mu)} =  \frac{c(\mu)}{m(\mu)}
\end{equation}
with arbitrary constant coefficients $c_2,\dots, c_{2N+n}$. Observe that the second highest coefficient of the polynomial $c(\mu)$ is zero.  

Our main result is that solutions of the integrable system generated by the commuting functions $F_0, \dots, F_{N-1}$ with $f(\mu)$ defined by \eqref{eq:Umu}  are naturally related to solutions  of   \eqref{eq:explform}.
More specifically, consider the zero level surface of the integrals 
$$
\mathcal X=\{F_0=\dots= F_{N-1}=0\} \subset T^*\R^N.
$$ 
and solutions located on it.
Let $x$ denote the time of the Hamiltonian flow of $H = -F_{0}$ and $t$ denote the time of the Hamiltonian  flow of  $F_{1}$. Since these flows commute,   we can obtain\footnote{solving  two ODEs, \eqref{eq:ODE1} and \eqref{eq:ODE2}}  their common solution  $(w_1(x,t),\dots,w_N(x,t); p_1(x,t),\dots, p_N(x,t))$ for any initial point located on $\mathcal X$.  Let us show how to  produce a solution $(u_1(x,t), \dots, u_n(x,t))$ of  \eqref{eq:explform} from  $(w_1(x,t), \cdots, w_{N}(x,t))$.

Take $L= L(u)$ which was used for the construction of our BKM IV system \eqref{eq:explform}. The transformation $w\mapsto u$  is given by the following algebraic condition: there exists a polynomial   $Q(\mu)$ of degree $\le 2N-1$ such that 
\begin{equation} \label{eq:polynomialrelation}
    \underbrace{\det\bigl(\mu\, \Id-L(u)\bigr)}_{\textrm{degree } n}\underbrace{\det\bigl(M(w) - \mu \,\Id\bigr)^2}_{\textrm{degree }  2N}- \underbrace{c(\mu)}_{\textrm{degree }   2N+n}  =  \underbrace{m(\mu)}_{\textrm{degree }  \le n-1} \underbrace{Q(\mu)}_{\textrm{degree } \le  2N-1}.
\end{equation}
Relation \eqref{eq:polynomialrelation}  essentially means that the polynomial in the left hand side is divisible by $m(\mu)$, and the result of division is a polynomial in $\mu$ of degree $\le  2N-1$. Note that the coefficient at $\mu_{2N+n}$ in the left hand side vanishes, so that the left hand side of \eqref{eq:polynomialrelation}  is a polynomial  in $\mu$ of degree $\le 2N+n-1$ whereas the polynomial in the left hand side has degree $\le 2N+n-2$.  Hence, \eqref{eq:polynomialrelation} 
is quite a nontrivial condition establishing certain correspondence between $u\in \R^n$ and $w\in\R^N$.

  \begin{Fact} \label{fact:3}
Relation \eqref{eq:polynomialrelation} uniquely and algorithmically determines the coefficients of the characteristic polynomial $\det\bigl(\mu\, \Id- L(u)\bigr)$ as rational functions of $w=(w_1,\dots,w_n)$.   Since $L(u)$ is differentially non-degenerate, this allows us to reconstruct $u=(u_1,\dots,u_n)$ from these coefficients, so that as a result,   \eqref{eq:polynomialrelation} defines a map $\mathcal R: \R^N(w) \to \R^n(u)$. 
  \end{Fact}

The following result is a straightforward corollary of Theorems \ref{t1} and \ref{t2}.

\begin{Theorem}
    Let  $w_1(x,t),...,w_N(x,t)$ be a solution of the finite-dimensional integrable system  constructed   above. Then 
    $(u_1(x,t),\dots, u_n(x,t))=\mathcal R(w_1(x,t),...,w_N(x,t))$
is a solution of the   BKM  IV system \eqref{eq:explform}.
\end{Theorem}

\subsection{ Solutions coming from the  finite-dimensional reduction } \label{sec:1.4}

In \S \ref{sec:1.3}, see \eqref{eq:polynomialrelation}, we explained how to produce solutions of \eqref{eq:explform} using solutions $w(x,t)$ 
of an explicit finite-dimensional Hamiltonian system.  In order to find solutions of finite-dimensional systems, one can use for example  standard numerical ODE solvers.  Fixing a polynomial $c(\mu)$ gives us explicit expressions for $F_0=-H, \dots,F_{N-1}$ as functions of $2N$ variables $w_1,\dots,w_N, p_1,\dots, p_N$.  Choose an initial point $\widehat w, \widehat p$ such that $F_0(\widehat w, \widehat p)=\cdots=F_{N-1}(\widehat w, \widehat p)=0$ and
 solve numerically, e.g. with Maple,  the Hamiltonian system  
 \begin{equation}\label{eq:ODE1}  
 \tfrac{\ddd }{\ddd x}w_i=  -\tfrac{\partial F_0 }{\partial p_i}  \  , \   \    \tfrac{\ddd }{\ddd x}p_i=  \tfrac{\partial F_0 }{\partial w_i}   \ \textrm{with $w(0)= \widehat w, p(0)= \widehat p$},
 \end{equation} 
to obtain a solution $\tilde w(x), \tilde p(x)$. Then, for any $x$ viewed now as a parameter, solve numerically the 
Hamiltonian system  
 \begin{equation}\label{eq:ODE2}  
 \tfrac{\ddd }{\ddd t}w_i=  \tfrac{\partial F_1 }{\partial p_i}  \  , \   \    \tfrac{\ddd }{\ddd t}p_i=  -\tfrac{\partial F_1 }{\partial w_i}   \ 
 \textrm{with $w(x,0)= \tilde w(x)$, $p(0,x)= \tilde p(x)$}.
 \end{equation} 
 We obtain  $w(x,t), p(x,t)$.   Plugging $w(x,t)$ in \eqref{eq:polynomialrelation}  and resolving, we obtain $u(x,t)$ which is a numerical solution of the initial BKM system \eqref{eq:explform}. Let us demonstrate how this method works. We start with the well-studied KdV case.

\begin{figure}[!tbp]
  \centering
  \begin{minipage}[b]{0.45\textwidth}
    \includegraphics[width=\textwidth]{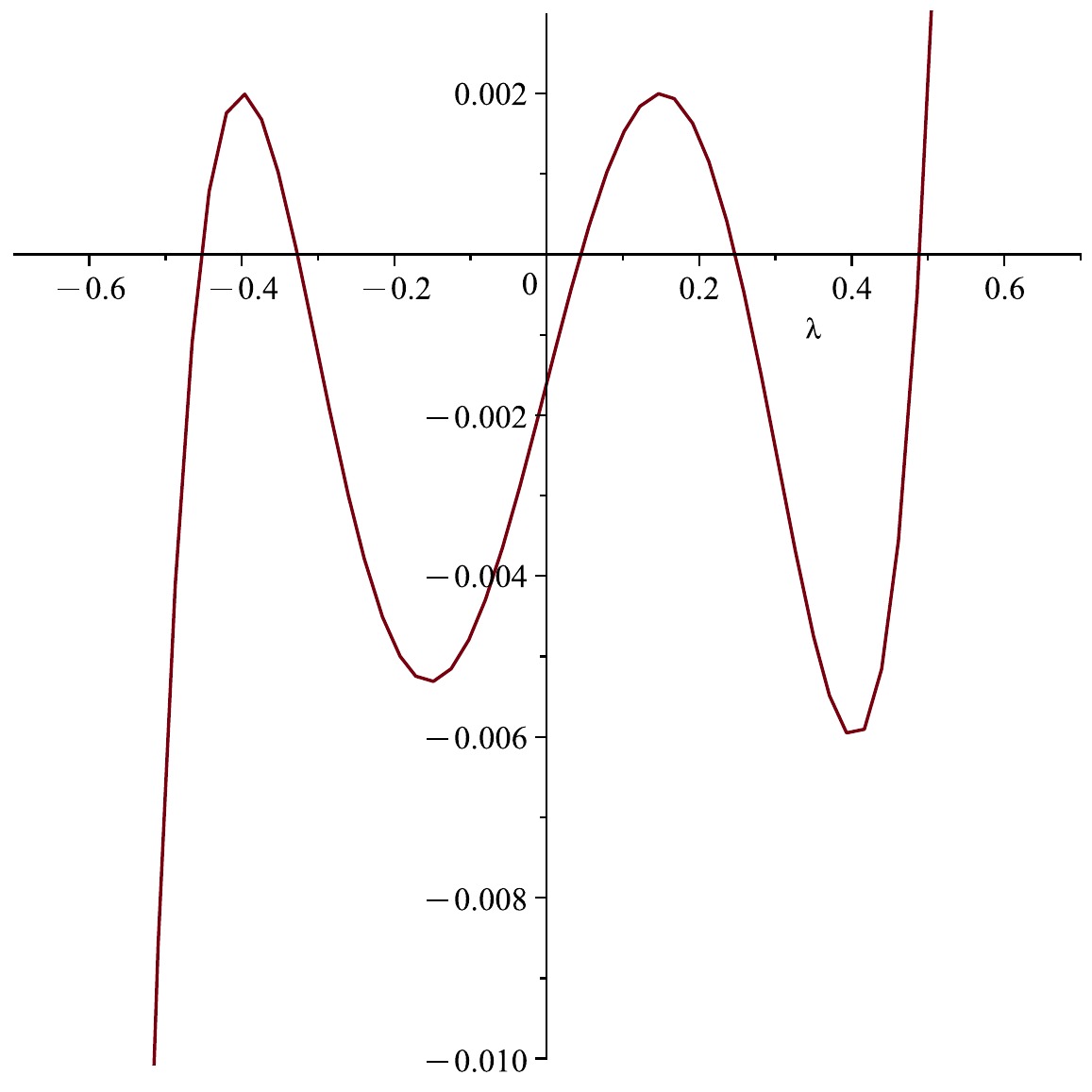}
       \caption{Polynomial $c(\mu)$ corresponding to cnoidal solutions of KdV.}
       \label{Fig:1}
  \end{minipage}
  \hfill
  \begin{minipage}[b]{0.45\textwidth}
   \includegraphics[width=\textwidth]{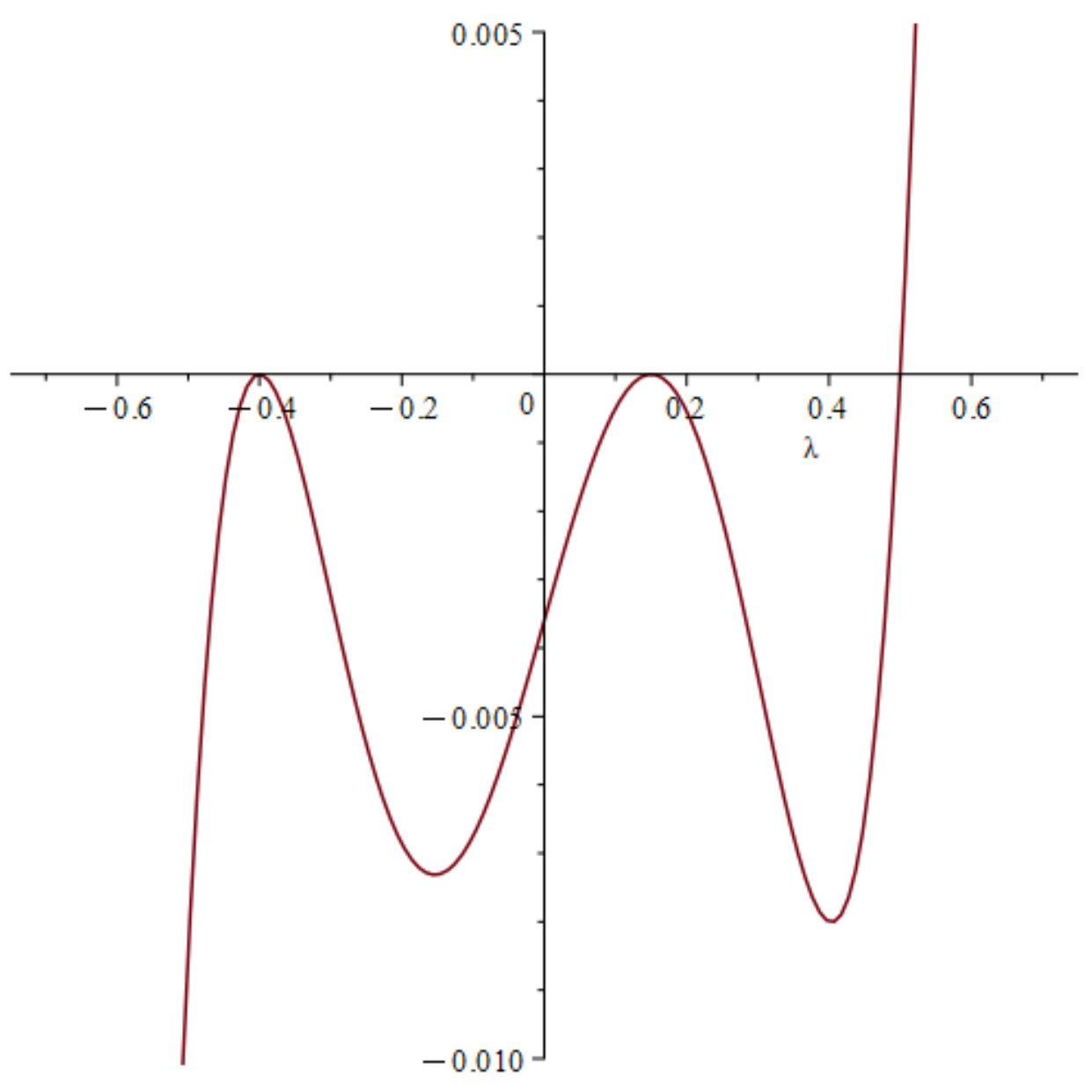}
       \caption{Polynomial $c(\mu)$ corresponding to two-solition solutions of KdV.}
       \label{Fig:2}
  \end{minipage}
\end{figure}

\begin{Ex}[Cnoidal and soliton solutions of KdV]  \label{Ex:4}  {\rm Take $n=1,N=2, m(\mu)= 1$.  Then,  $c(\mu)/m(\mu) =c(\mu)$ is a polynomial of  degree five, and we choose $c(\mu)$  such that it has   5 real roots $\widehat q_1<\cdots <\widehat q_5$, see  Fig. \ref{Fig:1}.  Next, take   $\widehat w_1 =-  (\widehat q_2+ \widehat q_4), \widehat w_2= \widehat q_2 \widehat q_4$  as the initial data. Numerically solving the ODEs \eqref{eq:ODE1}  and \eqref{eq:ODE2} and substituting the result  into \eqref{eq:polynomialrelation}, which  gives $u=2 w_1$ in the case 
$n=1$,  we obtain the ``cnoidal'' behaviour as one can see on \textrm{\url{https://youtu.be/NUr8D4ZDmmY}}. 
Note that the  corresponding solutions of the  finite-dimensional Hamiltonian system generated by $F_0 , F_1$ live  on a Liouville torus and behave  quasi-periodically. This   explains  a quasi-periodic behaviour of cnoidal solutions. 

To obtain a two-soliton solution, we  take the  polynomial   $c(\mu)  $  with two double roots, see    Fig. \ref{Fig:2}.    The animation of the behavior is on  \textrm{\url{https://youtu.be/Rdbt_Ez03r0}}. On this animation, we clearly see a two-soliton interaction.   
}\end{Ex}

%%%%%   eps  \to pdf   %%%%  need to return back

\begin{figure}[!tbp]
  \centering
  \begin{minipage}[b]{0.33\textwidth}
    \includegraphics[width=\textwidth]{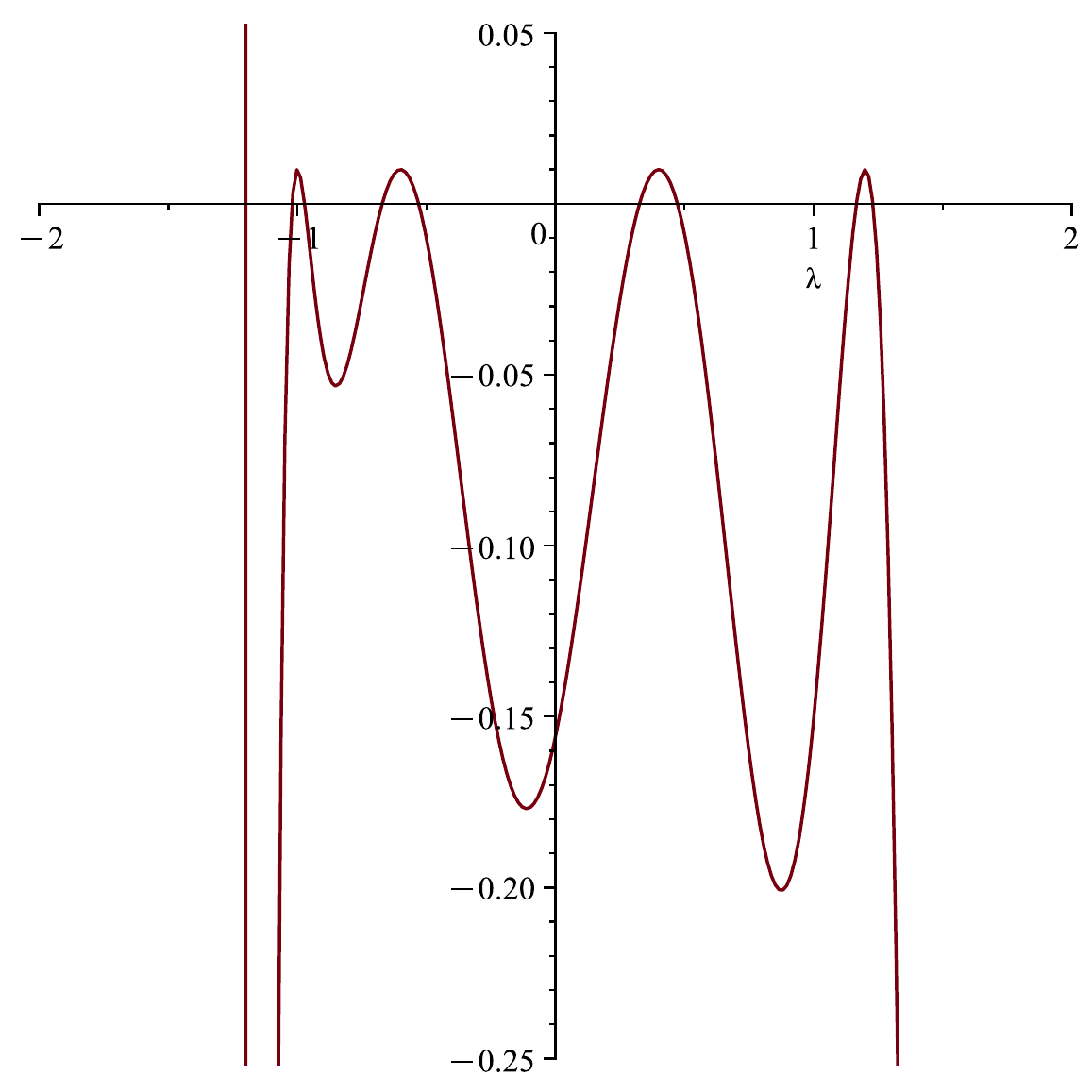}
       \caption{$c/m $ corresponding to cnoidal solutions of BKM.}
       \label{Fig:3}
  \end{minipage}
  \hfill
  \begin{minipage}[b]{0.33\textwidth}
   \includegraphics[width=\textwidth]{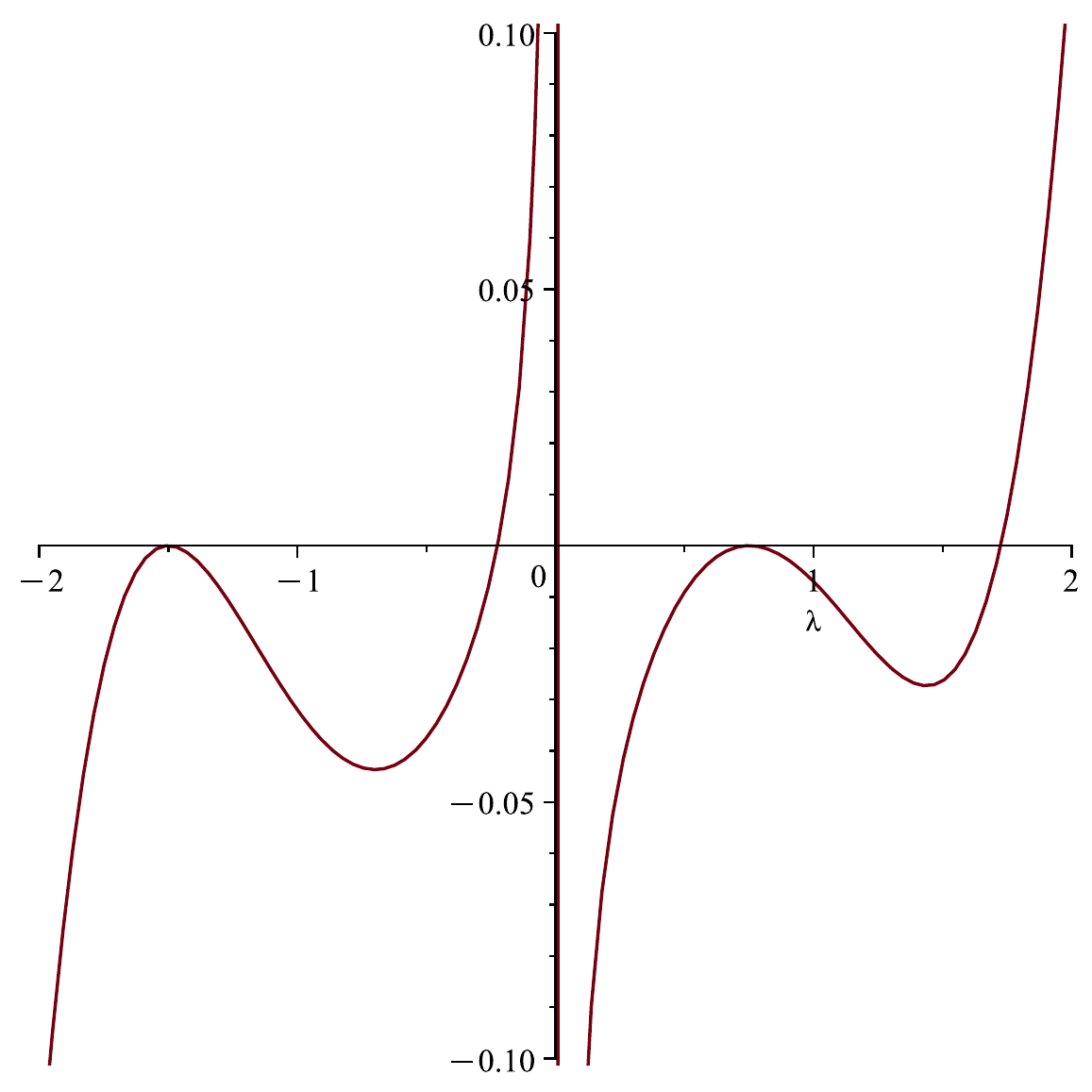}
       \caption{$c/m $ corresponding to a soliton ``loop'' solutions with $n=2, N=2$}
       \label{Fig:4}
  \end{minipage}\begin{minipage}[b]{0.33\textwidth}
    \includegraphics[width=\textwidth]{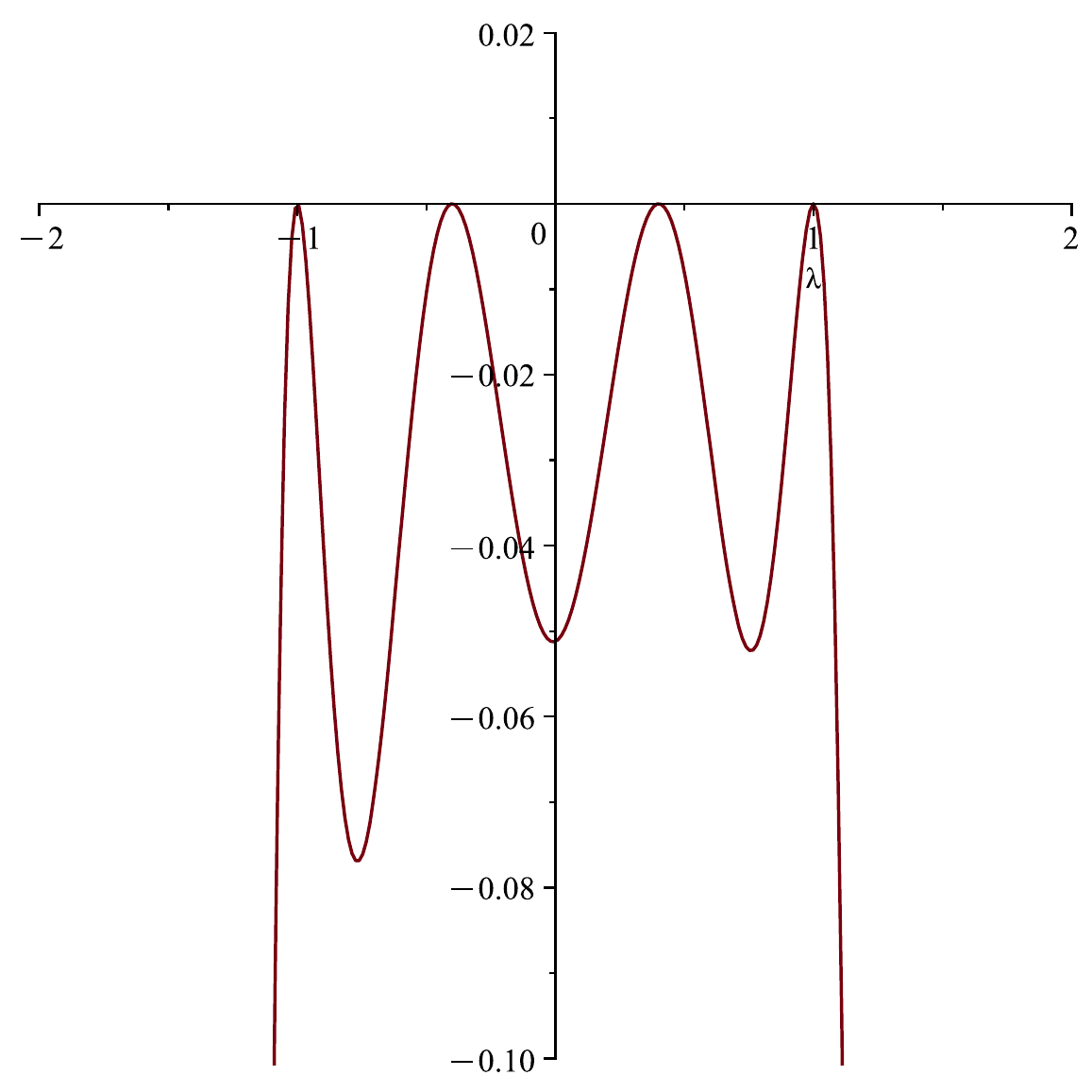}
       \caption{Polynomial $c(\mu) $ corresponding to soliton ``skipping rope'' solutions of KdV.}
       \label{Fig:5}
  \end{minipage}
\end{figure}

\begin{Ex}[Cnoidal and soliton  solutions of BKM IV with $n=2$]  \label{Ex:5}{\rm  A cnoidal solution of KdV is a finite-gap solution such that the corresponding  trajectory of the reduced system lies on a Liouville torus. This definition naturally extends to general BKM systems.    Take 
$c(\mu)/m(\mu)$  whose diagram looks as in Fig. \ref{Fig:3}, it corresponds to $n=2, \ N=3$.  For accurately chosen initial data, the behaviour of the solutions is on \textrm{\url{https://youtu.be/SzerRj2u18s}}. As in the KdV case, we clearly see quasi-periodic behaviour. 

In the KdV case, a soliton can be defined as a finite-gap solution which is asymptotically constant for $x\to \pm\infty$. In the KdV case, the  limits for $x\to +\infty$ and $x\to - \infty$ necessary coincide, in the BKM cases with $n>2$ they may be different constants. For the function $c(\mu)/m(\mu)$ shown in Fig. \ref{Fig:4}, the asymptotic values 
for $x\to \pm\infty$ coincide, so for every $t$ the  curve $x\mapsto (u_1(x,t), u_2(x,t))$ is a loop  with a fixed origin. The animation of the behavior is on  \textrm{\url{https://youtu.be/KRfOcUbxTgA}}.  For the function $c(\mu)/m(\mu)$ shown in  Fig.  \ref{Fig:5}, the asymptotic values 
for $x\to \pm\infty$ are different, so  the  curves $x\mapsto (u_1(x,t), u_2(x,t))$ connect two fixed points, see  animation on  \textrm{\url{https://youtu.be/50bWEScKhV8}}. 
}\end{Ex}

Let us now discuss analytical  ways to obtain solutions. It is convenient to pass to the coordinates $q_1,\dots ,q_N$ discussed in Fact \ref{fact:2}, so we assume that $(g, M) $ are given by \eqref{eq:g_and_M} and the integrals satisfy \eqref{eq:stackel}.   Such systems can be integrated by the method of separation of variables: it we denote by $t_0$ the time corresponding to the Hamiltonian system generated by  $F_0$, $t_1$ the time corresponding to the Hamiltonian system generated by  $F_1$ and so on, 
we obtain  
\begin{equation} \label{eq:separation}
    \begin{array}{ccc}
\sqrt{2}\,t_0 & =  & \sum_{i=1}^N \pm \int^{q_i} \frac{s^{N-1}} {\sqrt{c(s)/m(s)}}  ds \\
\sqrt{2}\, t_1  & = &   \sum_{i=1}^N \pm \int^{q_i} \frac{s^{N-2}} {\sqrt{c(s)/m(s)}} ds \\
&\vdots& \\ 
\sqrt{2}\, t_{N-1}  & = &  \sum_{i=1}^N \pm \int^{q_i} \frac{1} {\sqrt{c(s)/m(s)}} ds .
\end{array} 
\end{equation}
This is a system of explicit functional equations on $q_1, \dots, q_N $, depending on $t_0,\dots, t_{N-1}$ as parameters. Solving this system with respect to $q_i$  we obtain  $q_i(t_0,...,t_{N-1})$.   
Fixing parameters  $t_2,\dots t_{N-1}$,   replacing $t_0$ by $-x$ and denoting $t_1$ by $t$,
we obtain the evolution $q_i(x,t)$ which gives us, via \eqref{eq:polynomialrelation}, a  solution of the  initial KdV system. 
   
\begin{Remark} \label{Ex:6}{\rm
It is known that many solutions of KdV systems  can be found explicitly, in generalised special or  sometimes even in elementary  
functions. This is because the corresponding integrals in  \eqref{eq:separation} can  be solved. For example, for the two-solition solution $c(\mu)= (\mu-a)(\mu-b)^2(\mu-c)^2  $  with $a+ 2b +2c=0$, so the integrals  in \eqref{eq:separation} can be solved in elementary functions and the system \eqref{eq:separation} gives famous exact two-soliton solutions of KdV.   
   The situation with BKM systems with $n\ge 2$ is more complicated: although for some $c/m$ the integrals can be solved in elementary functions, solving \eqref{eq:separation} in elementary functions remains a non-trivial problem. We plan to study exact solutions for BKM systems in  a more general class of functions (e.g., generalised theta-functions) in the future. In certain cases, one can solve the finite-dimensional system in elementary functions, see   Example \ref{ex:last} below.     
}\end{Remark}

\begin{Ex} \label{ex:last} We consider the BKM IV system from Example  \ref{ex:KB},  with $n=2$ and  $m(\mu) = 1$.  We take $N=2$ and 
$c(\mu)= (\mu-1)^2\mu^2(\mu+1)^2.$  Such $c(\mu)$ corresponds to a  soliton-type solution.  
For an  appropriate choice of  the signs $\pm$ in \eqref{eq:separation},  and assuming $x=t_0, \, t= t_1$,  we obtain:
\begin{eqnarray} \sqrt{2}\, x& =& -\frac{\ln \! \left(q_{1}+1\right)}{2}+\frac{\ln \! \left(1-q_1\right)}{2}+\frac{\ln \! \left(q_{2}+1\right)}{2}-\frac{\ln \! \left(1-q_{2}\right)}{2}  \label{eqar:1}\\
\sqrt{2}\, t  &=& 
\frac{\ln \! \left(q_{1}+1\right)}{2}-\ln \! \left(-q_{1}\right)+\frac{\ln \! \left(1-q_{1}\right)}{2}-\frac{\ln \! \left(q_{2}+1\right)}{2}+\ln \! \left(q_{2}\right)-\frac{\ln \! \left(1-q_{2}\right)}{2}.\end{eqnarray}
The system can be solved for $q_1(t, x), q_2(t, x)$ in elementary functions: 
$$ 
q_{1} = 
-\frac{{    e}^{-   \sqrt{2}\, ( { x   } +     {t })} \left({    e}^{2  \sqrt{2}  \,  { x   } }-1\right)}{{    e}^{   \sqrt{2}\,( { x   } -     {t })}+{    e}^{-   \sqrt{2}\,( { x   } +    {t })}+2}
\, , \  \   q_{2} = 
\frac{{    e}^{2    \sqrt{2}\,  { x   } }-1}{2 \,{    e}^{  \sqrt{2}\,(  { x   } -     {t })}+{    e}^{2   \sqrt{2}\,  { x   }}+1}. 
$$
The corresponding $(u_1(t,x), u_2(t,x))$  solving  the Kaup-Boussinesq system  \eqref{eq:KB} with $m_0=1$  are 
related to  the above $q_1,q_2$ by \eqref{eq:polynomialrelation} which gives  
$u_1 = 2 q_{1}+2 q_{2} \ , \ \  {u_2} = 
3 q_{1}^{2}+4 q_{1} q_{2}+3 q_{2}^{2}-2.$
The animation of the behaviour is on \url{https://youtu.be/Wt6EN0Av8O0}.

\end{Ex}

 \subsection{Scheme of the proof}

 Our method constructs solutions via finite-dimensional reductions leading to an  integrable system on $\R^N$, where $N$ can be arbitrary large. Our reduction procedure consists of two steps, Theorems \ref{t1} and \ref{t2}, and its scheme is in Fig. \ref{Fig:6}.  First, we reduce our PDE system to another system of two differential equations on $\R^N$.  One of them, called {\it base equation},  is an overdetermined second order ODE system,  the other is a quasilinear PDE system that defines dynamics on the space of the solutions to the base equation.  This {\it first} reduced system is genetically related to the original PDE system and for this reason this intermediate step is important.  Theorem \ref{t2}   links the reduced system with an integrable Hamiltonian system on $T^*\R^N$ described in \S \ref{sec:1.3}.  The $t$-dynamics on the set of trajectories  is naturally given by one of the first integrals of the system. 
\begin{figure}
    \centering
    \includegraphics[width=0.5\linewidth]{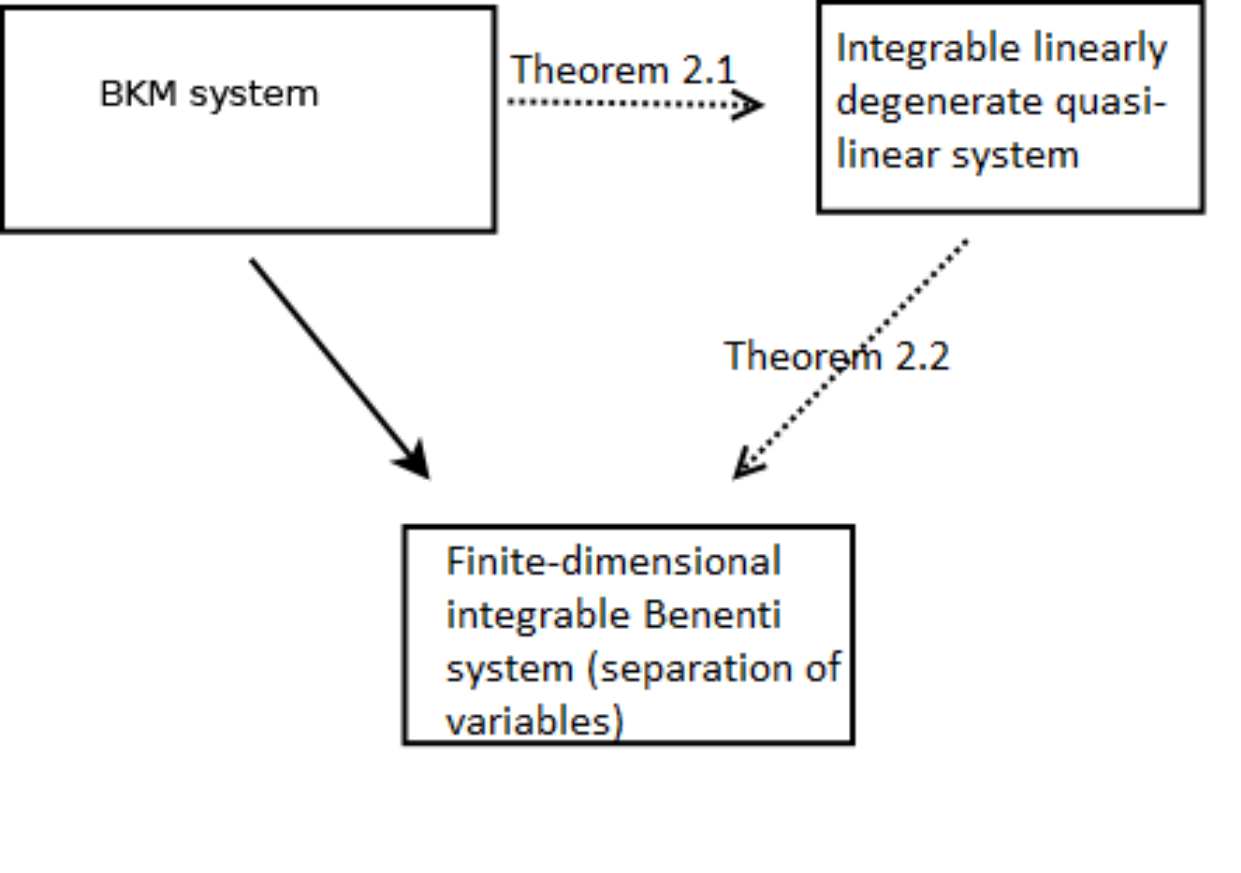}
    \caption{Scheme of the proof}
    \label{Fig:6}
\end{figure}

{\bf Acknowledgements.}  The authors are very grateful to
H.~Dullin,  E.~Ferapontov,  K.~Marciniak, A. Mironov, V.~Novikov, 
P.~Van der Kamp and  A.~Veselov for numerous discussions and explanations.  

A.B. was supported by the Ministry of Science and Higher Education of the Republic of Kazakhstan (grant No. AP23483476).
V.M. thanks the DFG (projects 455806247 and 529233771), and the ARC Discovery Programme DP210100951 for their support. A part of the work was done during the preworkshop and workshop on Nijenhuis Geometry and Integrable Systems at La Trobe University and the Matrix Institute. The participation of A.K. and V.M. at the workshop was supported by the Simons Foundation, and the participation of A.K. at the preworkshop was partially supported by the ARC Discovery Programme DP210100951.  A substantial part of this work was done when A.B. and V.M. took part in the Thematic Programme ``Geometry beyond Riemann: Curvature and Rigidity'' at the Erwin Schr\"odinger International Institute for Mathematics and Physics of the University of Vienna in the fall of 2023.   We thank ESI for the support, hospitality and fantastic research facilities.

%%%%%%%%%%%%%%%%%%%%%%%%%%%%%%%%%%%%%%%%%%%

\section{Ingredients of the construction and the main results}
\label{intro1}

%%%%%%%%%%%%%%%%%%%%%%%%%%%%%%%%%%%%%%%%%%%

The main subject of our paper is a series of multicomponent integrable PDE systems with a differential constraint constructed in  \cite{nijapp4}.  The systems of type I or II form a family of pairwise commuting evolutionary flows parametrised by $\lambda\in\R$.  Systems of type III and IV  correspond to $\lambda = \infty$. 

The main ingredients of the construction are a differentially non-degenerate Nijenhuis operator $L$ on $\R^n(u_1,\dots,u_n)$, an arbitrary polynomial  
$$
m(\mu)= m_0 + m_1 \mu + m_2 \mu^2 + \dots + m_n \mu^n
$$ 
of degree $\leq n$ with constant coefficients, and  the vector field $\zeta$ uniquely defined by the relation 
$$
    \mathcal L_\zeta \sigma(\mu, u) = m(\mu) - m_n \sigma(\mu, u),
$$
where 
$ \mathcal L_\zeta$ denotes the Lie derivative and $\sigma(\mu, u) = \det (\mu \operatorname{Id}-L)$
is the characteristic polynomial of  $L$.

For each  $\lambda\in\R$ (we may also consider $\lambda\in\mathbb C$), we define 
\begin{equation}\label{eq:BKM}
\begin{aligned}
    u_{t_\lambda} & = {q}_{xxx} \Bigl(L - \lambda \operatorname{Id}\Bigr)^{-1}\zeta + {q} \Bigl(L - \lambda \operatorname{Id}\Bigr)^{-1}u_x, \\
    1 & = m(\lambda) \Bigl( {q}_{xx} {q} - \tfrac{1}{2}\, {q}_x^2\Bigr) + \sigma(\lambda, u) {q}^2.
\end{aligned}    
\end{equation}
It is an evolutionary PDE with a single differential constraint (BKM I), which can be equally understood as a system of 
two PDEs on $n+1$  unknown functions $u_1(x,t_\lambda),..., u_n(x,t_\lambda), q(x, t_\lambda)$ of two variables.   If $\lambda$ is a root of $m(\cdot)$, then the constraint gives $q = \sigma(\lambda, u)^{-1/2}$ and hence can be ignored. The first equation of \eqref{eq:BKM} gives then an evolutionary system PDEs on $n$  unknown functions $u_1(x,t_\lambda),..., u_n(x,t_\lambda)$ of two variables with no constraint (BKM II):
$$
 u_{t_\lambda}  = \left(\tfrac{1}{\sqrt{ \sigma(\lambda, u)}} \right)_{xxx}\Bigl(L - \lambda \operatorname{Id}\Bigr)^{-1}\zeta + \tfrac{1}{\sqrt{ \sigma(\lambda, u)}} \Bigl(L - \lambda \operatorname{Id}\Bigr)^{-1}u_x.
$$  
In order to keep uniformity of the construction, we will not specially distinguish this case.

The parameter $\lambda$ in system \eqref{eq:BKM} is an arbitrary real or complex number.  However,  $\lambda = \infty$ still make sense, and the corresponding equations were introduced in \cite{nijapp4} as systems of type III and IV.   To obtain them from  \eqref{eq:BKM}, we replace $\lambda$ with $\lambda^{-1}$ and in the new system, after appropriate rescaling,  take the linear term in its Taylor expansion in $\lambda$ at the point $\lambda = 0$.  These operations lead to the following BKM III system:
\begin{equation}
\label{eq:BKM3}
\begin{aligned}
u_{t_\infty} &= q_{xxx} \zeta + (L + q\,\Id) u_x,\\
0 &= 2q +  m_n q_{xx} - \tr L.
\end{aligned}
\end{equation}
where all the ingredients have the same meaning as above in \eqref{eq:BKM} and  $m_n\in \R$ is the highest coefficient of the polynomial $m(\mu)= m_n \mu^n + \dots$    

If $m_n = 0$ (we may interpret this as a root at infinity), then the second relation in \eqref{eq:BKM} implies  $q = \frac{1}{2}\tr L$ and \eqref{eq:BKM}  becomes a usual multicomponent evolutionary equation  (BKM IV): 
$$
u_{t_\infty} = \tfrac{1}{2} (\tr L)_{xxx} \, \zeta + \left(L + \tfrac{1}{2} \tr L \cdot \Id\right)\, u_x.
$$
Although in terms of applications, the latter equation is very important,  we will consider it as just a special case of \eqref{eq:BKM3} without highlighting this case in our presentation.

Our goal is to describe a special class of solutions to \eqref{eq:BKM} and \eqref{eq:BKM3} via a finite-dimensional reduction. To that end, we choose an arbitrary natural number $N$ and introduce another system of two differential equations on $\R^N$.  Similar to systems \eqref{eq:BKM} and \eqref{eq:BKM3}, the construction of this system is based on a differentially non-degenerate Nijenhuis operator $M$  defined on $\R^N$.  If the coefficients $w_1, \dots, w_N$ of the characteristic polynomial
$$
\det(\mu\,\Id - M) = \mu^N + w_1 \mu^{N-1} + w_2\mu^{N-2} + \dots + w_N
$$
are taken to be coordinates on $\R^N$, then according to \cite{nij1}, $M$ is uniquely defined and has the standard companion form (see explicit formula \eqref{eq:M} below). 
The construction, however, is invariant and can be referred to any other coordinate system on $\R^N$. 

The first equation (which we call  {\it base equation}) is introduced as  
$$
m(\mu) \Bigl(w_{xx}(\mu) w(\mu) - \tfrac{1}{2} w_x^2(\mu)\Bigr) + {\rho}(\mu, w)\, w^2(\mu) = c(\mu), \quad \mbox{for all $\mu \in \R$},
$$
where 

\begin{itemize}
\item $w(\mu)=\mu^N + w_1\mu^{N-1} + w_2\mu^{N-2} + \dots + w_N$ is the characteristic polynomial of $M$ with $w_i=w_i(x)$ being unknown functions,  

\item $m(\mu)$ is the polynomial of degree $\le n$ with constant coefficients, same as in \eqref{eq:BKM}, 

\item $c(\mu)$ is an arbitrary monic polynomial of degree $2N + n$ with constant coefficients.
The polynomials $m(\mu)$ and $c(\mu)$ can be understood as parameters of the base equation. 

\item ${\rho}(\mu,w)$ is a monic polynomial of degree $n$ with coefficients depending on $w=(w_1,\dots, w_N)$ that depends on the choice of $c(\mu), m(\mu)$ and is uniquely defined by the following algebraic condition (see more details in Section \ref{2.2}): 

the polynomial  ${\rho}(\mu,w) w^2(\mu) - c(\mu)$  is divisible by $m(\mu)$ and moreover, the ratio
$$
\frac{{\rho}(\mu,w) w^2(\mu) - c(\mu)}{m(\mu)}
$$
is a polynomial of degree $\le 2N-1$\footnote{Notice that ${\rho}(\mu,w) w^2(\mu) - c(\mu)$ has degree $\le 2N+n -1$, since  by our assumptions ${\rho}(\mu,w) w^2(\mu)$ and $c(\mu)$ are monic polynomials of degree $2N+n$ so that, if we subtract one from the other, the highest degree terms cancel out.  This explains why the degree of the ratio ${\rho}(\mu,w) w^2(\mu) - c(\mu)$ and $m(\mu)$ is required to have degree $\le 2N-1$.}.
\end{itemize}

Using the above definition of ${\rho}(\mu, w)$, we can rewrite the base equation in the form
$$
w_{xx}(\mu) w(\mu) - \tfrac{1}{2} w_x^2(\mu) + \frac{{\rho}(\mu,w) w^2(\mu) - c(\mu)}{m(\mu)}=0, \quad \mbox{for all $\mu \in \R$},
$$
with the l.h.s. being a polynomial in $\mu$ of degree $2N+n-1$.  In particular, the base equation amounts to $2N$ second order ODEs on $N$ unknown functions $w_1(x),\dots, w_N(x)$, obtained by equating all the coefficients of this polynomial to zero.

The second equation is the quasilinear system $w_t = M_\lambda w_x$,  where $w=\begin{pmatrix} w_1 \\ \vdots \\ w_N\end{pmatrix}$ and 
\begin{equation}
\label{eq:Mlambda}
M_\lambda = \begin{cases}
\det(\lambda\Id - M) \cdot (M-\lambda\Id)^{-1}, & \mbox{if $\lambda \in \R$}, \\
 M -\tr M \cdot\Id , & \mbox{if $\lambda = \infty$.}
\end{cases}
\end{equation}

We now consider the base equation together with this quasilinear system as a system of differential equations for $N$ unknown functions $w_i(x,t)$, $i=1,\dots, N$: 
\begin{subequations}
\label{eq:BKMreduced}
\begin{align}
m(\mu) &\Bigl(w_{xx}(\mu) w(\mu) - \tfrac{1}{2} w_x^2(\mu)\Bigr) + {\rho}(\mu,w) w^2(\mu) = c(\mu), \quad \mbox{for all $\mu \in \R$}. 
\label{eq:1a} \\
w_t & = M_\lambda w_x, \label{eq:1b}
\end{align}
\end{subequations}
where $\lambda$ is either a real number or $\lambda = \infty$ depending on whether we treat the evolutionary PDE system \eqref{eq:BKM} (type I and II) or system \eqref{eq:BKM3} (type III and IV).

Our first theorem explains the relationship between PDE systems \eqref{eq:BKM}, \eqref{eq:BKM3} and system \eqref{eq:BKMreduced}.

\begin{Theorem}\label{t1} Let $\lambda\in \R$, then
every solution $w_1(t,x), \dots, w_N(t,x)$ to  \eqref{eq:BKMreduced} {\rm(}i.e., a common solution of the base equation \eqref{eq:1a} and the quasilinear system \eqref{eq:1b}{\rm)} leads to a solution $u(t, x), {q}(t,x)$ of the PDE system \eqref{eq:BKM}  via the relations:
$$
\begin{aligned}
 \sigma \bigl(\mu, u(t,x)\bigr)&= {\rho}\bigl(\mu, w\left( a  t,x\right)\bigr),\\
{q}(t,x) &= \frac{1}{a} \Bigl(\lambda^N +w_1\left(a  t,x\right)\lambda^{N-1}  + \dots +  w_N\left(a  t,x\right)\Bigr),
\end{aligned} 
$$
where $a=\sqrt{c(\lambda)}$.

Similarly, if $\lambda = \infty$, then every solution $w_1(t,x), \dots, w_N(t,x)$ to  \eqref{eq:BKMreduced} leads to a solution $u(t, x), {q}(t,x)$ of the PDE system \eqref{eq:BKM3}  via the relations:
$$
\sigma\bigl(\mu, u(t,x)\bigr) = \rho \bigl(\mu , w(t, x + \tfrac{c_1}{2} t)\bigr)\quad\mbox{and}\quad q(t,x) = w_1(t,x) -\tfrac{c_{1}}{2},
$$
where $c_1$ is the first non-trivial coefficient of the polynomial $c (\mu) =  \mu^{2N+n} + c_1\mu^{2N+n-1} + \dots + c_{2N+n}$.
\end{Theorem}

The second theorem explains the geometric meaning of equations \eqref{eq:BKMreduced}   as an integrable finite-dimensional system that describes the motion of a point on $\R^N$ endowed with some flat metric $g_0$ and potential $U_0$.  Both $g_0$ and $U_0$ are naturally related to the Nijenhuis operator $M$.  Recall that  in the coordinates $w_1,\dots,w_N$ (coefficients of the characteristic polynomial of $M$), the operator $M$ takes the first companion form ( see e.g. \cite[Theorem 4.4, formula (15)]{nij1})
\begin{equation}
\label{eq:M}
M = \begin{pmatrix}
-w_1 & 1 &   & & \\
-w_2 & 0  & 1 & & \\
\vdots & &\ddots &\ddots & \\
-w_{N-1}& & &0 & 1 \\
-w_N& & & &0 \\
\end{pmatrix}.
\end{equation}   

Then, in the same coordinates, the contravariant metric $g_0^{-1}$ is defined by  
\begin{equation}
\label{eq:g0}
g_0^{-1} = \begin{pmatrix}
& & & & 1 \\
& & & 1 & w_1\\
& &\iddots & \iddots & w_2 \\
& 1 &w_1 & \iddots &\vdots \\
1 & w_1 & w_2 &\dots & w_{N-1}\\
\end{pmatrix}.
\end{equation} 
This metric is geodesically compatible with $M$ in the sense of \cite{nijapp5, nij1}. Moreover, $g_0$ can be characterised as the unique (up to a constant factor) flat metric which is geodesically compatible with $M$ and is globally defined for all values of coordinates $w_1,\dots, w_N$ if we think of them as complex numbers\footnote{Any other flat metric $g$ that is geodesically compatible with $L$ can be written as $g= g_0 p(M)^{-1}$, where $p(\cdot)$ is a polynomial of degree $\le N$.  One can easily see that $g$ blows up at those points where one of the eigenvalues of $M$ is a root of $p(\cdot)$ so that the operator $p(M)$ is not invertible. To avoid this situation we have to take $p=\mathrm{const}$.}.

Recall that geodesically compatible $g_0$ and $M$ give rise to a natural dynamical system on $T\R^N$ that can be integrated by separation of variables.  In the context of Nijenhuis geometry, such a system can be described as follows.  For an arbitrary real analytic function $f(t)$, consider the following matrix relation
\begin{equation}
\label{eq:integrals}
f(M) = U_0 M^{N-1} + U_1 M^{N-2} + \dots + U_{N-1}\Id,
\end{equation} 
where the coefficients $U_0, \dots, U_{N-1}:\R^N \to \R$ are uniquely defined smooth functions. Then the geodesic flow of the metric $g_0$ with the potential $U_0$, i.e., the ODE system of the form
\begin{equation}
\label{eq:geodflow}
\nabla_{\dot w} \dot w - \operatorname{grad} U_0 = 0,
\end{equation}
is completely integrable and admits the following family of first integrals quadratic in velocities
\begin{equation}
\label{eq:Flambda}
F_\mu (w,\dot w) = \tfrac{1}{2} g_0 (M_\mu \dot w, \dot w) - V_\mu(w),   \quad \mu \in\R,
\end{equation}
where $M_\mu$ is defined by \eqref{eq:Mlambda} and the potentials $V_\mu:\R^N \to \R$ are defined as $V_\mu = U_0\mu^{N-1} + U_1\mu^{N-2} + \dots + U_{N-1}$.   For $\mu = \infty$,   we set
\begin{equation}
\label{eq:Finfty}
F_\infty (w,\dot w) = \tfrac{1}{2} g_0 (M_\infty \dot w, \dot w) - V_\infty(w),
\end{equation}  
with $M_\infty = M - \tr M\cdot \Id$ and $V_\infty =  - U_1$. 

Notice that $M_\mu$ and $V_\mu$ are both polynomials of degree $N-1$ in $\mu$ so that alternatively we can replace $M_\mu$ and $V_\mu$ by their coefficients to get $N$ independent first integrals. Also, one can easily check that
$\mu^{N-1} F_{\frac{1}{\mu}} =  - H  -  \mu F_\infty + \dots$  where dots denote higher order terms in $\mu$ and $H=\frac{1}{2} g_0(\dot w, \dot w) + U_0(w)$, so that $F_\infty$ indeed is a first integral of \eqref{eq:geodflow}.

After identifying $T\R^N$ with $T^*\R^N$ by means of $g_0$, the integrals  $F_\mu$'s commute w.r.t. the canonical Poisson bracket,  leading to Liouville integrability of \eqref{eq:geodflow}.  

\begin{Theorem}\label{t2}
\begin{itemize}
\item[\rm{(i)}] The base equation \eqref{eq:1a} is equivalent to the following systems of first order ODEs
\begin{equation}
\label{eq:c3}
F_\mu (w,\dot w) = 0,   \quad \mbox{for all $\mu \in\R$},
\end{equation}
where $F_\mu$ and  $V_\mu$ are defined by \eqref{eq:Flambda}  and  \eqref{eq:integrals} with $f(t)=\frac{c(t)}{m(t)}$.  In other words, the solutions $w=w(x)$ of the base equation are the trajectories of the dynamical system \eqref{eq:geodflow}
located on the common level surface of the first integrals $F_\lambda$ 
\begin{equation}
\label{eq:intsurface}
\mathcal X = \{ (w,\dot w) \in T\R^N~|~ F_\mu (w,\dot w) = 0,   \quad \mbox{for all $\mu \in\R$}\} \subset T\R^N.
\end{equation}
In particular, the solutions $w(x) = \bigl( w_1(x),\dots,w_N(x)\bigr)$ of the base equation form an $N$-parameter family of curves.
\item[\rm{(ii)}]  The family of solutions to the base equation  \eqref{eq:1a}  is invariant w.r.t. the evolutionary flow defined by the quasilinear system  \eqref{eq:1b}.
In particular,  system \eqref{eq:BKMreduced} admits a solution $w(t,x)=\bigl( w_1(t,x), \dots, w_N(t,x)  \bigr)$ for any initial condition 
$w(0,x)$ being a solution to the base equation  \eqref{eq:1a}.    
 
\item[\rm{(iii)}]  Solutions $w(t,x)=\bigl( w_1(t,x), \dots, w_N(t,x)  \bigr)$ of \eqref{eq:BKMreduced} can be obtained by integrating 
two commuting Hamiltonian flows $\Phi^x_H$ and $\Phi^t_{F_\lambda}$ on the cotangent bundle $T^*\R^N$ (identified with $T\R^N$ by means of $g_0$), where 
$$
H(w,p) = \frac{1}{2} g_0^{-1}(p,p) + U_0(w), \quad F_\lambda(w,p)=\frac{1}{2} g_0^{-1}(M_\lambda^* p, p) - V_\lambda(w)
$$   
are the Hamiltonian and the first integral of the geodesic flow \eqref{eq:geodflow} with the same $\lambda$ as in the quasilinear system  \eqref{eq:1b}.  More  precisely,   let 
$$
\bigl(  w(t,x), p(t,x)  \bigr) = \Phi^x_H \circ \Phi^t_{F_\lambda} (\widehat w, \widehat p), \quad (t,x)\in\R^2,
$$
be the orbit of the Hamiltonian $\R^2$-action generated by $H$ and $F_\lambda$ with
an initial point $(\widehat w, \widehat p)\in \mathcal X \subset T^*\R^N\simeq T\R^N$ on the integral surface \eqref{eq:intsurface}.  Then $w(t,x)$ is a solution of \eqref{eq:BKMreduced} and every solution of \eqref{eq:BKMreduced} can be obtained in this way.
\end{itemize}
\end{Theorem}

%%%%%%%%%%%%%%%%%%%%%%%%%%%%%%%%

We conclude this section with a step-by-step algorithm that can be used for constructing solutions of \eqref{eq:BKM}  and \eqref{eq:BKM3} by using Theorems \ref{t1} and \ref{t2}, i.e., by solving a certain integrable Hamiltonian ODE system with $N$ degrees of freedom.

\begin{itemize}

\item[Step 0.]   Choosing parameters of the system.  These parameters are an arbitrary polynomial $m(\mu) = m_0 + m_1\mu + \dots + m_n \mu^n$ of degree $\le n$  (in particular,  the coefficient $m_n$ can equal zero) and  a number $\lambda\in \R$ for equation \eqref{eq:BKM} or $\lambda = \infty$ for equation \eqref{eq:BKM3}.

\item[Step 1.]  Explicit form of the equation.   

As already noticed,  system \eqref{eq:BKM} has an invariant meaning and can be written in any coordinate system.   The explanation below is given in the coordinate system $u^1,\dots, u^n$, in which $L$ has the following form (called first companion form \cite{nij3}): 
$$
L = \begin{pmatrix}
     u^1 & 1 & 0 & \dots & 0  \\
     u^2 & 0 & 1 & \dots & 0  \\
     \vdots & \vdots &\ddots & \ddots & \\
     u^{n - 1} & 0 & \dots & 0 & 1 \\
     u^n & 0 & \dots & 0 & 0
\end{pmatrix}
$$

  In this setting, the other ingredients of the BKM equations   \eqref{eq:BKM}  and \eqref{eq:BKM3} (of type I and III respectively)  are as follows:
$$
\begin{aligned}
m(\mu) &= m_0 + m_1\mu + \dots + m_n \mu^n, \\
\sigma(\mu, u) &= \mu^n - u_1\mu^{n-1} - \ldots - u_{n-1} \mu \ - \, u_n,   \\
\zeta & = - \begin{pmatrix} m_n u_1 + m_{n-1} \\ m_n u_2 + m_{n-2} \\ \dots \\ m_n u_n + m_0 \end{pmatrix}\\ 
\end{aligned}
$$

\item[Step 2.]   Choosing a monic polynomial of degree $2N + n$  (this polynomial serves as a parameter of the reduced  system \eqref{eq:BKMreduced})
$$
c(\mu) = \mu^{2N+n} + c_1 \mu^{2N+n-1} + c_2 \mu^{2N+n-2} + \dots   
$$
 
\item[Step 3.]   Define  the map  $\mathcal R : \R^N(w_1, \dots, w_N) \to \R^n(u_1,\dots, u_n)$ which will send solutions $w(t,x)$ of the reduced system  \eqref{eq:BKMreduced} to solutions $u(t,x)$ of the original PDE system \eqref{eq:BKM}. In the setting of Theorem \ref{t1},  this map amounts to the relation $\sigma(\mu, u) = \rho(\mu, w)$, i.e., 
establishes a correspondence between $w$ and $u$ by resolving this relation with respect to $u$ in our special coordinate system chosen in Step 1. 

One of the main ingredients of the base equation is the polynomial 
$$
\rho (\mu, w) = \mu^n + \rho_1(w) \mu^{n-1} + \dots + \rho_{n-1}(w)\mu + \rho_n(w),
$$
whose coefficients are certain functions of $w_1, \dots, w_N$ that can be explicitly found from the condition that the ratio
$$
Q(\mu, u) = \frac{\rho (\mu, w) w^2(\mu) - c(\mu)}{m(\mu)}
$$ 
is a polynomial in $\mu$ of degree $\le 2N-1$.
In the generic case, when $m(\mu)$ has $n$ distinct roots $\lambda_1,\dots, \lambda_n$ (and these roots are explicitly given), the above condition implies $\rho (\lambda_i, w) w^2(\lambda_i) - c(\lambda_i)=0$, that is, $\rho (\lambda_i, w) = \dfrac{c(\lambda_i)}{w^2(\lambda_i)}$. Hence,  the polynomial $\rho_i$ can be reconstructed by using the Lagrange interpolating polynomial:
$$
\rho(\mu, w) = \mu^n + \sum_{i=1}^n \left( \dfrac{c(\lambda_i)}{w^2(\lambda_i)}  - \lambda_i^n \right) \prod_{s\ne i} \frac{\mu - \lambda_s}{\lambda_i - \lambda_s} 
$$
where $w(\mu) = \mu^N + w_1 \mu^{N-1} + \dots + w_N$.   
 
In the general case, when $m(\mu)$ has multiple roots  $\lambda_1,\dots, \lambda_s$ of multiplicities $k_1, \dots, k_s$   (with $k_1 + \dots + k_s = \deg m(\mu) \le n$),   the polynomial $\rho$ can be reconstructed from the following $d=\deg m(\mu)$ linear conditions at the roots of $m(\mu)$:
 $$
 \rho(\lambda_i) = h(\lambda_i),  \dots, \rho'(\lambda_i)=  h'(\lambda_i),\dots , \rho^{(k_i-1)}(\lambda_i)=h^{(k_i-1)}(\lambda_i), \quad i=1,\dots, s,
 $$
 where $h(\mu)=\frac{c(\mu)}{w^2(\mu)}$ and the derivative is taken with respect to $\mu$.  If $d<n$, then we may interpret $\lambda=\infty$ as an additional root of $m$  of multiplicity $n-d$ to obtain $n-d$ additional conditions of the form
 $$
 \bar \rho'(0) = \bar h'(0), \ \bar \rho''(0) = \bar h''(0), \ \dots \ , \   \bar \rho^{(n-d)}(0) = \bar h^{(n-d)}(0), 
 $$
 where 
 $$
 \begin{aligned}
 \bar\rho(\mu) &= \mu^n \rho(\tfrac{1}{\mu}) = 1 + \rho_1\mu + \rho_2\mu^2 + \dots \\
 \bar h(\mu) &= \mu^n \frac{c(\frac{1}{\mu})}{ w^2 (\frac{1}{\mu})} = \frac{1 + c_1\mu +c_2\mu^2+\dots}{ (1 + w_1\mu + w_2\mu^2 + \dots)^2}
 \end{aligned}
 $$
 Resolving these conditions leads to explicit description of the coefficients $\rho_1(w), \dots, \rho_n(w)$ of the polynomial $\rho(\mu,w)$.
 
 Of course, one can also find these coefficients by formally resolving the polynomial relation
 $$
  \rho (\mu, w) w^2(\mu) - m(\mu) Q(\mu,w) =  c(\mu)   
 $$
 which amounts to a system of  $2N + n$ non-homogeneous linear equations on $2N + n$ unknown coefficients of the polynomials
 $$
 \rho  = \mu^n + \rho_1 \mu^{n-1} + \rho_2 \mu^{n-2} + \dots\quad\mbox{and}\quad Q= q_1 \mu^{2N-1} + q_2 \mu^{2N-2} + \dots
 $$
 
 The coefficients $\rho_1(w), \dots, \rho_n(w)$ of the polynomial $\rho(\mu, w)$ so obtained allow us to define the map
 $$
 \begin{aligned}
 \mathcal R   :   \R^N(w_1, \dots, w_N) & \to \R^n(u_1,\dots, u_n),    \\
 w = (w_1,\dots,w_N) & \overset{\mathcal R}{\mapsto} \mathcal (u_1, \dots, u_n) = \bigl(-\rho_1(w), \dots, -\rho_n(w)\bigr).
 \end{aligned}
 $$

Note that the polynomial $\rho(\mu, w)$  is uniquely defined by the parameters of the above constructions, i.e., polynomials $c(\mu)$ and $m(\mu)$.  However,  one can easily notice that $\rho(\mu, w)$ remains unchanged if we replace $c(\mu)$ by 
$c(\mu) +  m(\mu) q(\mu)$, where $q(\mu)$ is an arbitrary polynomial of degree $\le 2N-1$ with constant coefficients.

\item[Step 4.]  Commuting flows (ODE systems) on $T \R^N \simeq T^*\R^N$.    

Here our goal is to resolve  (analytically or numerically)  an integrable Hamiltonian system associated with another Nijenhuis operator $M$.  This system can be written and then solved in various coordinate systems. We will provide formulas related to the first companion coordinate system $(w_1,\dots, w_N)$ in which $M$ takes form \eqref{eq:M}.

We consider the natural system on $\R^N(w_1,\dots, w^N)$ whose kinetic energy is defined by the (contravariant) metric $g_0$ given by \eqref{eq:g0}
and potential  $U_0=U_0(w)$  defined from the following matrix relation
$$
c(M) \bigl(m(M)\bigr)^{-1} = U_0(w) M^{N-1} + U_1(w) M^{N-2} + \dots + U_{N-1}(w) \Id
$$
Notice that in the coordinates $w^1,\dots, w^n$, the coefficients $U_0, \dots, U_{N-1}$ of the matrix polynomial in the right hand side are exactly the elements of the last column of the matrix $c(M) \bigl(m(M)\bigr)^{-1}$ in the left hand side. 

Next we consider the Hamiltonian system with the Hamiltonian
$$
H(w, p) = \frac{1}{2} g_0^{-1}(p, p) + U_0(w) =\frac{1}{2} \sum_{i, j = 1}^N (g_0)^{ij} p_i p_j + U_0(w)  
$$
with $g_0$ and $U_0$ defined above  (note that $(g_0)^{ij}$ are the entries of $g_0^{-1}$ defined by \eqref{eq:g0} so that in the Hamiltonian setting we do not need to invert this matrix).

We will also need an integral $F(w,p)$  of this  Hamiltonian system. Namely, in the case of equation \eqref{eq:BKM} we take 
$$
F_\lambda (w,p) = \frac{1}{2} g_0^{-1}(M_\lambda^* p, p) - V_\lambda (w),   
$$
where $M_\lambda = \det (\lambda\Id - M) \cdot (M - \lambda \Id)^{-1}$ and $V_\lambda = U_0 \lambda^{N-1} + U_1\lambda^{N-2} + \dots + U_{N-1}$.   Recall that $M_\lambda^* : T^*_w \R^N \to T^*_w \R^N$ is the dual operator to $M_\lambda : T_w \R^N \to  T_w \R^N$ (which, in terms of matrices, is equivalent to transposition).

Similarly, in the case of equation \eqref{eq:BKM3} we formally set $\lambda = \infty$ and take
$$
F_\infty(w,p) =  \frac{1}{2} g_0^{-1}\bigl( (M - \tr M\cdot \Id )^* p, p\bigr) + U_1 (w).
$$

The functions $H(w,p)$ and $F_\lambda(w,p)$ ($\lambda\in\R \cup \{\infty\}$) commute with respect to the canonical Poisson bracket on $T^*\R^N$  and, therefore, generate  a Hamiltonian $\R^2$-action on $T^*\R^N$.

\item[Step 5.]   Finding orbits of the Hamiltonian $\R^2$-action generated by $H$ and $F_\lambda$ ($\lambda\in\R \cup \{\infty\}$).

This step is equivalent to simultaneously solving the Hamiltonian systems generated by $H$ and $F_\lambda$:
\begin{equation}
\label{eq:twoflows}
\left\{ \begin{aligned}  
\frac{\dd w_\alpha}{\dd x} & = \frac{\partial H}{\partial p_\alpha} \\ 
\frac{\dd p_\alpha}{\dd x} & = - \frac{\partial H}{\partial w_\alpha}
\end{aligned}\right.
\quad \mbox{and} \quad 
\left\{ \begin{aligned}  
\frac{\dd w_\alpha}{\dd t} & = \frac{\partial F_\lambda}{\partial p_\alpha} \\ 
\frac{\dd p_\alpha}{\dd t} & = -\frac{\partial F_\lambda}{\partial w_\alpha}
\end{aligned}\right.
\end{equation}

This integration can be done in quadratures because $H$ and $F_\lambda$ are, in natural sense, included into a Liouville integrable system with $N$ degrees of freedom (see Theorem \ref{t2}).  This way leads to a quite complicated analytic formulas involving hyperelliptic integrals.  For some good choice of parameters  (i.e., coefficients of the polynomials $m(\mu)$ and $c(\mu)$) these formulas simplifies and one can expect to find explicit solutions in elementary functions.  

On the other hand, since we deal with a system of ODEs,  one can integrate \eqref{eq:twoflows} numerically. To that end,  we 
\begin{itemize}
\item choose an arbitrary initial condition $(\widehat w, \widehat p)$,  
\item solve (numerically) the first system to find the integral curve
$(w(x), p(x))$ of the Hamiltonian flow of $H$   (such that $w(0)= \widehat w$ and $p(0) = \widehat p$) and  
\item taking $(w(x), p(x))$ as an initial condition,  solve  (numerically) the second system to find the integral trajectory $(w(t,x), p(t,x))$ of the Hamiltonian flow of $F_\lambda$ (such that $w(0, x) = w(x)$, $p(0,x) = p(x)$).   
\end{itemize}

Notice that according to Theorem \ref{t2},  we should consider the solutions located at the common zero level of the integrals $F_\lambda$ (see \eqref{eq:c3}).  If we take an arbitrary initial condition, then  $F_\lambda (\widehat w, \widehat p) =a_0 \lambda^{N-1} + a_1 \lambda^{N-2} + \dots + a_{N-1}\ne 0$ for some constants $a_i$.  We can, however, easily `repair' this situation by introducing  a new polynomial $c(\lambda)$  at Step 2.  Namely, 
$$
c_{\mathrm{new}} (\lambda) = c(\lambda) + m(\lambda) F_\lambda (\widehat w, \widehat p).
$$ 
This minor change will result in the shift  $(F_{\mathrm{new}})_\lambda = F_\lambda - F_\lambda (\widehat w, \widehat p)$, so that the new integrals will simultaneously vanish at the initial point and we return to the setting of Theorem \ref{t2}.
Notice that this change, does not affect the polynomial $\rho(\mu, w)$ and $\mathcal R$-mapping from Step 3.

\item[Step 6.]   Final step.  Finding solutions to the BKM systems \eqref{eq:BKM} and \eqref{eq:BKM3}.

Take the solution $(w(t,x), p(t,x))$  of \eqref{eq:twoflows} found in Step 5.
The corresponding solution $u(t,x)$ of \eqref{eq:BKM} is obtained from it by applying the mapping $\mathcal R : \R^N (w_1,\dots, w_N) \to \R^n(u_1,\dots,u_n)$ from Step 3:
$$
u(t, x ) = (u_1(t,x), \dots, u_N(t,x)) =\mathcal R\bigl( w_1(at,x), \dots, w_N(at,x)  \bigr), \quad\mbox{with } a = \sqrt{c_{\mathrm{new}}(\lambda)},
$$ 
for equation \eqref{eq:BKM}  (Type I and II) and
$$
u(t, x ) = (u_1(t,x), \dots, u_N(t,x)) =\mathcal R\bigl( w_1(t,x+\tfrac{c_1}{2} t), \dots, w_N(t,x+\tfrac{c_1}{2} t)  \bigr).
$$ 
for equation \eqref{eq:BKM3}  (Type III and IV).

 If necessary, we can also reconstruct the function $q(t,x)$:
$$
q(t,x) = \frac{1}{a} \Bigl(\lambda^N + w_1\left(a  t,x\right)\lambda^{N-1} + w_2\left(a  t,x\right)\lambda^{N-2} + \dots + w_N\left(a  t,x\right)\Bigr), \quad\mbox{where } a = \sqrt{c_{\mathrm{new}}(\lambda)}.
$$
for equation \eqref{eq:BKM}  (Type I and II) and
$$
q(t,x) = w_1(t,x) - \frac{c_1}{2}.
$$
for equation \eqref{eq:BKM3}  (Type III and IV).

\end{itemize}

%%%%%%%%%%%%%%%%%%%%%%%
%%%%%%%%%%%%%%%%%%%%%%%
%%%%%%%%%%%%%%%%%%%%%%%
%%%%%%%%%%%%%%%%%%%%%%%
%%%%%%%%%%%%%%%%%%%%%%%
%%%%%%%%%%%%%%%%%%%%%%%
%%%%%%%%%%%%%%%%%%%%%%%

\section{Proof of Theorem \ref{t1}}

\subsection{Nijenhuis operators, their conservation laws and integrable quasilinear systems}\label{first}

We start with a couple of basic formulas. The construction below gives a family of integrable quasilinear systems based on a Nijenhuis operator $L$ and one of its conservation laws $f$.  It is an equivalent version of the cohomological construction  of integrable hierarchies of hydrodynamic type due to P.\,Lorenzoni and F.\,Magri \cite{ml}.   We need some differential identities related to it.

\begin{Proposition}\label{p1}
Let $M$ be a Nijenhuis operator on  $\R^N(v^1, \dots, v^N)$,   $\ddd f$ be its conservation law and  $w(\mu)=w(v^1, \dots, v^N, \mu)$ a function on $\R^n$ depending on $\mu$ as a parameter and satisfying the relation
\begin{equation}\label{r1}
(M - \mu \operatorname{Id})^* \, \ddd w(\mu) = w(\mu) \,\ddd f \quad \mbox{for all $\mu\in\R$}.
\end{equation}
For a fixed parameter $\lambda\in\R$, consider the quasilinear system
\begin{equation}\label{eqmain}
    v^i_t = (M_\lambda)^i_q v^q_x, \quad i = 1, \dots, N, \quad \mbox{where $M_\lambda = w(\lambda) (M - \lambda \operatorname{Id})^{-1}$},
\end{equation}
and take an arbitrary solution  $v^i(t, x), i = 1, \dots, N$  of \eqref{eqmain}. Then the function 
$$
w(t, x, \mu) = w(v^1(t, x), \dots, v^N(t, x), \mu)
$$ 
satisfies the formula
    $$
\begin{cases}    \partial_{t} w(t, x, \mu) = \frac{1}{\mu - \lambda} \Bigl( w_x(t, x, \mu) w(t, x, \lambda) - w_x(t, x, \lambda) w(t, x, \mu)\Bigr), \quad &\mbox{for $\mu \ne \lambda$}, \\
  \partial_{t} w(t, x, \lambda) = w'_x(t, x, \lambda) w(t, x, \lambda) - w_x(t, x, \lambda) w'(t, x, \lambda), 
\quad &\mbox{for $\mu = \lambda$}, 
    \end{cases}
    $$
 where $w'$ stands for the derivative with respect to $\mu$.
\end{Proposition}

\begin{proof}
We start with 
\begin{Lemma}\label{lm0}
In the assumptions of Proposition \ref{p1}, for $\lambda \neq \mu$ the following identity holds
\begin{equation}\label{id1}
M^*_\lambda \ddd w(\mu) = \frac{1}{\mu - \lambda} \Bigl( w(\lambda) \ddd w(\mu) - w(\mu) \ddd w(\lambda)\Bigr).    
\end{equation} 
Here $\ddd w(\lambda) = \Big(\pd{w(\lambda)}{v^1}, \dots, \pd{w(\lambda)}{v^N}\Big)$. 
\end{Lemma}

\begin{proof}
%%%%%%%%%%%%%%%%%%%%%%%%%%%

Notice that \eqref{r1} can be equivalently rewritten as ${M_\mu^*}^{-1} \ddd w(\mu) =  \ddd f$. Since this identity holds for any $\mu$, we get
$$
{M_\lambda^*}^{-1} \ddd w(\lambda) = {M_\mu^*}^{-1}  \ddd w(\mu).
$$
Using the algebraic relation   $w(\mu){M_\mu^*}^{-1} = w(\lambda) {M_\lambda^*}^{-1} - (\mu - \lambda)\Id$, we get
$$
M_\lambda^{-1} \ddd w(\lambda) = {M_\mu^*}^{-1} \ddd w(\mu) = 
\frac{w(\lambda)}{w(\mu)}{M_\lambda^*}^{-1} \ddd w(\mu) - (\mu - \lambda)\frac{1}{w(\mu)}\ddd w(\mu).
$$
Multiplying this identity by $w(\mu) M_\lambda^*$ gives
$$
w(\lambda) \ddd w(\lambda) =  w(\lambda) \ddd w(\mu) - (\mu - \lambda) M^*_\lambda \, \ddd w(\mu),
$$
which is equivalent to \eqref{id1}. \end{proof}

To complete the proof of Proposition \ref{p1},  we apply Lemma \ref{lm0} in the following sequence of relations:
$$
\begin{aligned}
\partial_t w(t, x, \mu) & = \partial_t w(v^1(t, x), \dots, v^N(t, x), \mu) = \pd{w(\mu)}{v^i} v^i_t = \pd{w(\mu)}{v^i}(M_\lambda)^i_q v^q_x = \\
& = (M^*_\lambda \ddd w(\mu))_q  v^q_x = \frac{1}{\mu - \lambda} \left( w(t, x, \lambda) \pd{w(\mu)}{v^q} - w(t, x, \mu) \pd{w(\lambda)}{v^q}\right) v^q_x = \\
& = \frac{1}{\mu - \lambda} \Bigl( w(t, x, \lambda) w_x(t, x, \mu) - w(t, x, \mu) w_x(t, x, \lambda)\Bigr). 
\end{aligned}
$$
The second part of the formula is obtained by taking the limit as  $\mu \to \lambda$. 
\end{proof}

Let us introduce the following family of equations (depending on $\lambda$ as a parameter)
\begin{equation}\label{soliton}
\partial_{t_\lambda} w(\mu) = \frac{1}{\mu - \lambda} \Bigl(w_x(\mu) w(\lambda) - w(\mu) w_x(\lambda)\Bigr).    
\end{equation}
Notice that for every value of the parameter $\lambda$, a solution to this equation is a function of three variables $w(t_\lambda, x, \mu)$. As $\lambda$ is considered fixed  through the proof, we will write $t$ instead of $t_\lambda$.  

Note that \eqref{soliton} is not a PDE: the right hand side includes the values of $w$ and $w_x$ taken for the value of parameter different from that in the right hand side (i.e. for $\lambda$ but not $\mu$). Setting $\mu = \lambda + \epsilon$ and expanding the r.h.s. in powers of $\epsilon$, we get the infinite series of evolutionary equations. The initial condition of \eqref{soliton} is a function of two-variables $v(x, \mu) = w(0, x, \mu)$. The existence of the solution for certain initial conditions in analytic category can be obtained from the infinite-dimensional version of Cauchy-Kovalevskaya theorem in \cite{zub}. 

Equation  \eqref{soliton}  (in a slightly different but equivalent form)  was introduced by A.\,Shabat in \cite{sh1, sh2} as {\it universal solitonic equation}. Proposition \ref{p1} can be understood as follows: using Nijenhuis operators and solving the quasilinear system \eqref{eqmain}, one can construct a large family of  solutions to the universal solitonic equation.  In the context of this paper, the following example is crucial.

For an arbitrary Nijenhuis operator $M$, take $f =   \operatorname{tr} M$. Then the  
$$
w(v^1, \dots, v^N, \mu) =  \det (\mu \operatorname{Id} - M)
$$
satisfies \eqref{r1} and we obtain the following  
\begin{Corollary}\label{cor:1}
Let $M_\lambda =   \det (\lambda \operatorname{Id} - M) (M - \lambda \operatorname{Id})^{-1}$. Then the function 
$w(\mu) = w(u, \mu)= \det (\mu \operatorname{Id} - M(u))$ satisfies the relation
$
M^*_\lambda \ddd w(\mu) = \frac{1}{\mu-\lambda} \Bigl( w(\lambda) \ddd w(\mu) - w(\mu) \ddd w(\lambda)\Bigr),   
$
and for any solution to the quasilinear system $u_t = M_\lambda u_x$,  the function $w(\mu, t, x) = 
 \det \Bigl(\mu \operatorname{Id} - M\bigl(u(t,x)\bigr)\Bigr)$ satisfies the universal solitonic equation \eqref{soliton}.
\end{Corollary}
%%%%%%%%%%%%
%%%%%%%%%%%%
%%%%%%%%%%%%

To prove the second part of Theorem \ref{t1},  we need to modify Corollary \ref{cor:1} by including the case $\lambda = \infty$.
To that end, we introduce the operator
$\bar M_\lambda = \det (\Id - \lambda M) (\lambda M-\Id)^{-1}$ and set $\bar w(\lambda) = \det (\Id - \lambda M) =
1 + w_1\lambda + \dots + w_N \lambda^N$.  In this notation, the formula from Proposition \ref{p1} becomes
\begin{equation}
\label{eq:3.1}
\partial_{t}  w(t, x, \mu) = \frac{1}{\mu\lambda-1} \Bigl(  w_x(t, x, \mu) \bar w(t, x, \lambda) - \bar w_x(t, x, \lambda)  w(t, x, \mu)\Bigr), \quad \mbox{for $\mu \lambda\ne 1$}
\end{equation}
where $\bar w (t,x,\lambda) = \det \bigl(\Id - \lambda\cdot M(v(t,x))\bigr)$ and  $v^i(t,x)$'s satisfy  $v^i_t = (\bar M_\lambda)^i_q v^q_x$.

We now use the expansion $\bar M_\lambda = - \Id - \lambda M_\infty + \dots$ with 
$$
M_\infty=-{\frac{\dd}{\dd\lambda}|}_{\lambda = 0} \bar M_\lambda=M - \tr M\cdot \Id.
$$   
By differentiating  the r.h.s. of  \eqref{eq:3.1}  w.r.t. $\lambda$ and then substituting $\lambda = 0$, we come to the following conclusion.

\begin{Proposition}
\label{prop:infty1}
Let   $v(t,x) = \bigl(v^1(t,x), \dots , v^N(t,x)\bigr)$ solve the quasilinear system  
$$
v^i_t = (M_\infty)^i_q v^q_x,
$$
and $w (t,x,\mu) = \det \bigl(\mu\Id - M(v(t,x))\bigr)$.
Then
\begin{equation}
\label{eq:3.2}
\partial_{t}  w(t, x, \mu) = \mu  w_x(t, x, \mu) +
 w_x(t, x, \mu) w_1 -   w(t, x, \mu)(w_1)_x.
\end{equation}
\end{Proposition}

%%%%%%%%%%%%
%%%%%%%%%%%%
%%%%%%%%%%%%
%%%%%%%%%%%%

\subsection{\texorpdfstring{Properties of ${\rho}(t, x, \mu)$ }{Properties of the function  rho}}\label{2.2}

In Section \ref{intro1}, we introduced the function $\rho(\mu,w)$ appearing in the base equation \eqref{eq:1a} by using the following algebraic condition:   $\rho(\mu,w)$ is a monic polynomial in $\mu$ with coefficients depending on $w=(w_1,\dots,w_N)$ such that 
$\rho(\mu,w) w^2(\mu) - c(\lambda)$ is divisible by $m(\lambda)$ and the ratio of these polynomials is a polynomial of degree $\le 2N-1$.   

The fact that such a polynomial exists and is unique follows from a simple algebraic statement.

\begin{Lemma}
\label{lem:alg}  
Let  $r(\lambda)$ and $c(\lambda)$ be monic polynomials of degree  
$m$ and $n+m$ respectively and $m(\lambda)$ be a polynomial of degree $\le n$.  Assume that $r(\lambda)$ and $m(\lambda)$ are relatively prime, i.e., have no common roots. Then there exist a unique monic polynomial ${\rho}(\lambda)$ of degree $n$ and a unique polynomial $\tau(\lambda)$ of degree $\le m-1$ such that
$$
c(\lambda) = {\rho}(\lambda) r(\lambda) - \tau(\lambda) m(\lambda).
$$
Moreover,  the coefficients of ${\rho}(\lambda)$ and  $\tau(\lambda)$ are rational functions of the coefficients of $m$, $r$ and $c$ and these functions are smooth as soon as $m$ and $r$ have no common roots.  
\end{Lemma}
\begin{proof}
It is easy to see  that the problem reduces to solving a system of $n+m$ non-homogeneous linear equations on $n+m$ unknown coefficients of ${\rho}$ and $\tau$.   Therefore the statement of Lemma is equivalent to the fact that determinant of the matrix of this system is different from zero if and only if $r$ and $m$ have no common roots.  In this case, the coefficients of $\rho$ and $\tau$ can be found by Cramer's rule and, therefore, are rational functions of the coefficients of $m$, $r$ and $c$ as stated. 

Assume that the determinant is zero, then the corresponding homogeneous system admits non-zero solutions (and vice versa), i.e., there exist polynomials $\alpha(\lambda)$ and $\beta(\lambda)$ with $\deg \beta(\lambda)<\deg r(\lambda) = m$ such that 
$$
\alpha(\lambda) r(\lambda) - \beta(\lambda) m(\lambda) = 0, 
$$ 
In particular,  $\alpha(\lambda) = \dfrac{\beta(\lambda)m(\lambda)}{r(\lambda)}$. Since where $\deg \beta(\lambda)<\deg r(\lambda)$, then such a situation is possible if and only if $r(\lambda)$ and $m(\lambda)$ have at least one common root.  \end{proof}

In our case, $m = 2N$ and $r(\lambda) = w^2(\lambda)$.   Notice that $w^2(\lambda)$ is a polynomial with variable coefficients depending on $w_1,\dots, w_N$.  Hence, the coefficients of ${\rho}$ will be rational functions of $w_1,\dots, w_n$ and will be smooth unless $m(\lambda)$ and $w(\lambda) = \det (\lambda \Id - M)$ have common roots.  At those points, where one of the eigenvalues of $M$ is a root of $m(\mu)$,   the coefficients of ${\rho}(\mu, w)$ are expected to have poles.

Thus,  let ${\rho}(\mu, w)$ be the polynomial with coefficients depending on $w_1, \dots, w_N$ such that 
\begin{equation}
\label{eq:forrho}
\frac{{\rho} (\mu, w) w^2 (\mu) - c(\mu)}{m(\mu)}  =  Q(\mu, w),
\end{equation}
where $Q(\mu, w)$ is a polynomial in $\mu$ of degree $\le 2N-1$.  We now assume that 
$$
w(t,x,\mu) = \mu^N + w_1(t,x) \mu^{N-1} + \dots + w_N(t,x)
$$ 
satisfies the universal solitonic equation \eqref{soliton}, that is,  
$$
\partial_t w(\mu) = \frac{1}{\mu - \lambda} \Bigl(w_x(\mu) w(\lambda) - w(\mu) w_x(\lambda)\Bigr).    
$$
The next proposition shows that ${\rho} (t,x,\mu) = {\rho} (\mu, w(t,x))$ also satisfies a certain relation similar to  \eqref{soliton}.  

\begin{Proposition}\label{p3}
  Let $w(t,x,\mu)$ satisfy   \eqref{soliton}, then      
\begin{equation}\label{eq1}
    \begin{aligned}
        \partial_t {\rho}(t, x, \mu) = \frac{1}{\mu - \lambda} \Bigg[ & 2 {\rho}(t, x, \mu) w_x(t, x, \lambda) + {\rho}_x(t, x, \mu) w(t, x, \lambda) - \\
        & - \frac{m(\mu)}{m(\lambda)}\Bigl( 2 {\rho}(t, x, \lambda) w_x(t, x, \lambda) + {\rho}_x(t, x, \lambda) w(t, x, \lambda)\Bigr)
        \Bigg].
    \end{aligned}  
    \end{equation}
    
 Moreover, if in addition $w(t,x,\mu)$ satisfies the base equation \eqref{eq:1a}, then 
 \begin{equation}\label{eq1mod2}
        \partial_t {\rho}(t, x, \mu) = \frac{1}{\mu - \lambda} \Bigg[  2 {\rho}(t, x, \mu) w_x(t, x, \lambda) + {\rho}_x(t, x, \mu) w(t, x, \lambda)  + m(\mu) w_{xxx}(t,x,\lambda)        \Bigg].
    \end{equation}

\end{Proposition}

\begin{proof}
We first prove this formula in the case when $m(\mu)$ is a polynomial of degree $n$ with simple roots $\mu_1,\dots,\mu_n$. To shorten the notation in this proof, we use ${\rho}(\mu) = {\rho}(t,x,\mu)$ and $w(\mu) =w(t,x,\mu)$.

Recall that ${\rho}(\mu)$ is a monic polynomial of degree $n$ in $\mu$,  hence its derivative $\partial_t {\rho}(\mu)$ is a polynomial 
of degree $n-1$  with coefficients depending on $t$ and $x$.  Notice that the right hand side of \eqref{eq1} satisfies this condition.  Indeed, the expression in square brackets is a polynomial in $\mu$ of degree $n$  that vanishes if $\mu = \lambda$, hence division by $\mu - \lambda$ gives a polynomial of degree $n-1$ as required. 

To prove that the polynomials in the right and left hand sides coincide,  it suffices to verify this fact in $n$ distinct points. We will do it at the roots $\mu_i$'s of the polynomial $m(\mu)$.  
In view of \eqref{eq:forrho},   $m(\mu_i)=0$ implies ${\rho} (\mu_i) w^2 (\mu_i) - c(\mu_i) = 0$, i.e., ${\rho}(\mu_i) = \frac{c(\mu_i)}{w^2(\mu_i)}$. Hence,  
$$
\begin{aligned}
\partial_t {\rho}(\mu_i) & = \partial_t \left(\frac{c(\mu_i)}{w^2(\mu_i)}\right) =  - \frac{2 c(\mu_i)}{w^3(\mu_i)} \, w_t(\mu_i) 
= - \frac{2 c(\mu_i)}{w^3(\mu_i)} \cdot  \frac{1}{\mu_i - \lambda}  \Bigl( w_x(\mu_i) w(\lambda) - w(\mu_i) w_x(\lambda) \Bigr) = \\
& = \frac{1}{\mu_i - \lambda} \Bigg(2 \left(\frac{c(\mu_i)}{w^2(\mu_i)}\right) w_x(\lambda) + 
\left(- \frac{2 c(\mu_i)}{w^3(\mu_i)} w_x(\mu_i)\right) w(\lambda) \Bigg) \\ &= \frac{1}{\mu_i - \lambda} \Bigl(2 {\rho}(\mu_i) w_x(\lambda) + {\rho}_x (\mu_i) w(\lambda) \Bigr).
\end{aligned}
$$
It remains to notice that the right hand side of \eqref{eq1} gives the same result for $\mu=\mu_i$,  completing the proof in the generic case when $m(\mu)$ has $n$ simple roots. 

The case of $m(\mu)$ with multiple roots or of degree smaller that $n$ can be treated in a similar way.  However, one can use the  continuity argument instead.  Indeed,  the left and right hand sides of \eqref{eq1} both depend smoothly on the coefficients of $m(\mu)$.  If  \eqref{eq1} holds for an open dense subset in the space of these parameters, it holds identically for all of them. 

Formally speaking, in \eqref{eq1} we have to assume $m(\lambda)\ne 0$.  However,  one can notice that 
$\frac{1}{m(\lambda)} (2 {\rho}(t, x, \lambda) w_x(t, x, \lambda) + {\rho}_x(t, x, \lambda) w(t, x, \lambda) = \frac{Q_x(\lambda, w(t,x))}{w(t,x,\lambda)}$.  Hence, \eqref{eq1} can be rewritten in the following equivalent way which makes sense even if $\lambda$ is a root of $m$:
\begin{equation}\label{eq1mod}    
        \partial_t {\rho}(t, x, \mu) = \frac{1}{\mu - \lambda} \Bigg[  2 {\rho}(t, x, \mu) w_x(t, x, \lambda) + {\rho}_x(t, x, \mu) w(t, x, \lambda) - 
         m(\mu) \frac{Q_x(\lambda, w(t,x))}{w(t,x,\lambda)}
        \Bigg].
    \end{equation}

Finally, assume that $w(\mu)=w(t,x,\mu)$ satisfies the base equation, that is,
$$
w_{xx}(\mu) w(\mu) - \frac{1}{2} w_x(\mu)^2 + Q(\mu, w) = 0
$$ 
Differentiating this relation with respect to $x$ and dividing by $w(\mu)$, we get 
$$
w_{xxx}(\mu) = - \frac{Q_x(\lambda, w(t,x))}{w(t,x,\lambda)}
$$
Substituting it into \eqref{eq1mod} gives \eqref{eq1mod2}, which completes the proof. 
\end{proof}

For the second part of  Theorem \ref{t1},  we need the following analog of Proposition \ref{p3}, which can be proved in a similar way.

\begin{Proposition}\label{prop:infty2}
 Let $w(t,x,\mu)$ satisfy \eqref{eq:3.2} and $\rho(t,x,\mu) = \rho(\mu, w(t,x))$. Then
\begin{equation}\label{eq:rhoinfty}
        \partial_t  {\rho}(t, x, \mu) =    2  {\rho}(t, x, \mu) (w_1)_x +  {\rho}_x(t, x, \mu)
        \left( \mu + w_1\right) - \frac{m(\mu)}{m_n}\bigl(2(w_1)_x + (\rho_1)_x\bigr),
      \end{equation}
 where $\rho_1 = \rho_1(t,x)$ is the first non-trivial coefficient of the monic polynomial $\rho (t, x, \mu) = 
 \mu^n + \rho_1(t,x) \mu^{n-1} + \dots$    and  $w_1 = w_1(t,x)$ is the first non-trivial coefficient of the monic polynomial
$w (t, x, \mu) = 
 \mu^N + w_1(t,x) \mu^{N-1} + \dots$.
 
 Moreover, if in addition $w(t,x,\mu)$ satisfies the base equation  \eqref{eq:1a}, then
\begin{equation}\label{eq:rhoinftymod}
        \partial_t  {\rho}(t, x, \mu) =    2  {\rho}(t, x, \mu) (w_1)_x +  {\rho}_x(t, x, \mu)
        \left( \mu + w_1\right) + m(\mu) (w_1)_{xxx}.
      \end{equation}

\end{Proposition}

%%%%%%%%%
%%%%%%%%%
%%%%%%%%%
%%%%%%%%%
%%%%%%%%%
%%%%%%%%%

\subsection{Completing the proof of Theorem \ref{t1}}

Recall that \eqref{eq:BKM} is a one-parameter family of multicomponent evolutionary PDE systems with a differential constraint:  
\begin{equation}\label{bkm_b}
\begin{aligned}
 u_t & = {q}_{xxx}  \Big(L - \lambda \operatorname{Id}\Big)^{-1} \zeta + {q}  \Big(L - \lambda \operatorname{Id}\Big)^{-1} u_x, \\
 1 & = m(\lambda) \Big( {q}_{xx} {q} - \tfrac{1}{2} {q}_x^2 \Big) + \sigma(\lambda) \, {q}^2.
\end{aligned}    
\end{equation}
where $\lambda\in\R$ parametrises systems within the family\footnote{As shown in \cite{nijapp4}, the evolutionary flows related to different values of $\lambda$ pairwise commute and admit infinitely many common conservation laws, thus leading to integrability of \eqref{bkm_b}.} and, in what follows, we think of $\lambda$ as a fixed parameter. 

This system is defined by the choice 
of a differentially non-degenerate Nijenhuis operator $L$ on $\R^n(u^1,\dots,u^n)$ and a non-zero polynomial  $m(\mu) = m_n \mu^n + m_{n-1}  \mu^{n - 1} + \dots + m_0$  of degree $\le n$.  Two other ingredients, $\sigma(\mu)$ and $\zeta$, are defined in terms of $L$ and $m(\mu)$ as follows:
$$
\sigma(\mu)=\sigma(u^1,\dots, u^n,\mu) = \det \bigl(\mu\,\Id - L(u)\bigr),
$$ 
and $\zeta$ is the vector field on $\R^n$ satisfying the identity
\begin{equation}
\label{eq:Lzeta} 
\mathcal L_\zeta \sigma(\mu) = m(\mu) - m_n \sigma(\mu).
\end{equation}
The unknown functions are $u(t,x) =\begin{pmatrix}  u^1(t,x) \\ \vdots \\ u^n(t,x)    \end{pmatrix}$  and ${q}(t,x)$.

Our goal is to provide a method for constructing explicit solutions of \eqref{bkm_b}.   Since the coefficients of the characteristic polynomial $\sigma(\mu)$ are independent and can be taken as new coordinates,    \eqref{bkm_b} can be rewritten as an evolutionary PDE for $\sigma(\mu)$. 

\begin{Proposition}\label{ps_b}
Assume that $L, \sigma(\mu), \zeta, m(\mu)$ are defined as above. The following system 
    \begin{equation}\label{bkm1_b}
    \begin{aligned}
    \partial_t \sigma(\mu) & = \frac{1}{\mu - \lambda} \Bigl( m(\mu) {q}_{xxx}  + 2 \sigma(\mu) {q}_x + \sigma_x(\mu) {q} \Bigr), \\
    1 & = m(\lambda) \Big( {q}_{xx} {q} - \tfrac{1}{2} {q}_x^2 \Big) + \sigma(\lambda) {q}^2,    
    \end{aligned}
    \end{equation}
is equivalent to \eqref{bkm_b}. 
\end{Proposition}

\begin{Remark} For $\mu = \lambda$, the first relation of \eqref{bkm1_b} still makes sense and can be written as
    \begin{equation}\label{bkm1_2_b}
    \begin{aligned}
    \partial_t \sigma(\lambda) & = m'(\lambda) {q}_{xxx}  + 2 \sigma'(\lambda) {q}_x + \sigma'_x(\lambda) {q}, \\
    1 & = m(\lambda) \Big( {q}_{xx} {q} - \tfrac{1}{2} {q}_x^2\Big) + \sigma(\lambda) {q}^2,    
    \end{aligned}
    \end{equation}
    where $'$ stands for the derivative with respect to $\mu$.
\end{Remark}

\begin{proof}
We prove \eqref{bkm1_b} for $\mu \neq \lambda$ and then \eqref{bkm1_2_b} holds by continuity. Notice that Corollary \ref{cor:1} can be applied to any Nijenhuis operator so that we are allowed to replace $M$ with $L$ and $w(\mu)$ with $\sigma(\mu)$ (as $L$ and  $\sigma(\mu)$ satisfy the assumptions of Proposition \ref{p1}). Hence, we obtain  
\begin{equation}\label{ds1}
{(L - \lambda \operatorname{Id})^*}^{-1} \ddd \sigma(\mu) = \frac{1}{\mu - \lambda} \left(\ddd \sigma(\mu) - \frac{\sigma(\mu)}{\sigma(\lambda)} \ddd \sigma(\lambda)\right).
\end{equation}
Here $\ddd \sigma(\mu), \ddd \sigma(\lambda)$ are the differentials w.r.t. the coordinates $u^1,\dots, u^n$. 

Assume that $u^i(t, x), {q}(t, x)$ is a solution of \eqref{bkm_b} and set $\sigma( t, x,\mu) = \sigma(u^1(t, x), \dots, u^n(t, x), \mu)$. Then due to \eqref{eq:Lzeta} and \eqref{ds1},  the following holds  
\begin{equation}\label{ds2}
\begin{aligned}
\partial_t \sigma(t, x, \mu) & = \langle \ddd \sigma(\mu), u_t \rangle = \langle {(L-\lambda\,\Id)^*}^{-1} \ddd \sigma(\mu),  q_{xxx}\zeta + q u_x\rangle =
\\ & =\frac{{q}_{xxx}  }{\mu - \lambda} \left( m(\mu) - \frac{\sigma(\mu)}{\sigma(\lambda)} m(\lambda) \right) + \frac{{q}}{\mu - \lambda} \left( \sigma_x(\mu) - \frac{\sigma(\mu)}{\sigma(\lambda)} \sigma_x(\lambda)\right) = \\
& = \frac{1}{\mu - \lambda} \left(m(\mu) {q}_{xxx}  + \sigma_x(\mu) {q} - \frac{m(\lambda) {q}_{xxx}  + \sigma_x(\lambda) {q}}{\sigma(\lambda)} \sigma(\mu) \right).
\end{aligned}    
\end{equation}
We differentiate the second equation of \eqref{bkm_b} in $x$, divide by ${q}$ and get
$$
0 = m(\lambda) {q}_{xxx}  + 2 \sigma(\lambda) {q}_x + \sigma_x(\lambda) {q}.
$$
This equation is rewritten as
\begin{equation}\label{sm}
    - 2 {q}_{x}  = \frac{m(\lambda) {q}_{xxx}  + \sigma_x(\lambda) {q}}{\sigma(\lambda)}. 
\end{equation}
Substituting it into \eqref{ds2} we get the first equation of \eqref{bkm1_b}. Thus, \eqref{bkm_b} implies \eqref{bkm1_b}.

Now take a solution $\sigma(t, x, \mu), {q}(t, x, \mu)$ of \eqref{bkm1_b}. For the differentially non-degenerate Nijenhuis operator $L$, one can take the local coordinates to be the values of its characteristic polynomial $\sigma(\mu)$ at certain points $\mu_1, \dots, \mu_n$.  In terms of these coordinates $\sigma^i = \sigma(\mu_i)$,  the system \eqref{bkm1_b}  reads
\begin{equation}\label{bkm2}
\begin{aligned}
 \partial_{t_\lambda} \sigma^i & = \frac{1}{\mu_i - \lambda} \Bigl( a_i {q}_{xxx} +  2 \sigma^i {q}_x + \sigma^i_x {q}\Bigr), \quad i = 1, \dots, n, \\
1 & = m(\lambda) \Bigl({q}_{xx} {q} - \tfrac{1}{2} {q}_x^2\Bigr) + \sigma(\lambda) {q}^2,  
\end{aligned}
\end{equation}
where $a_i = m(\mu_i)$. In these coordinates, the vector field $\zeta$ takes the form $\zeta = (a_1 - m_0 \sigma^1, \dots, a_n - m_0 \sigma^n)^\top$. Next, \eqref{sm} yields the formula for $2{q}_x$  (notice that  \eqref{sm} was derived from the common second equation of \eqref{bkm_b} of \eqref{bkm1_b}). Substituting it into the equation above, we rewrite \eqref{bkm1_b} in the form
\begin{equation}\label{bkm3}
\begin{aligned}
 \partial_{t_\lambda} \sigma^i & = \frac{{q}_{xxx} }{\mu_i - \lambda}\left(a_i - \frac{m(\lambda)}{\sigma(\lambda)} \sigma^i\right) +  \frac{{q}}{\mu_i - \lambda}\left(\sigma^i_x - \sigma^i \frac{\sigma_x(\lambda)}{\sigma(\lambda)}\right), \quad i = 1, \dots, n, \\
 1 & = m(\lambda) \Big( {q}_{xx} {q} - \tfrac{1}{2} {q}_x^2\Big) + \sigma(\lambda) {q}^2, 
\end{aligned}
\end{equation}
Using \eqref{ds1} we get
$$
\begin{aligned}
 & \mathcal L_{(L - \lambda \operatorname{Id})^{-1} \zeta} \sigma^i = \frac{1}{\mu_i - \lambda} \left( a_i - m_0 \sigma^i - \frac{m(\lambda) - m_0 \sigma(\lambda)} {\sigma(\lambda)} \sigma_i  \right), \\
 & \langle \ddd \sigma^i, (\operatorname{Id} - \lambda L)^{-1} u_x\rangle = \frac{1}{\mu_i - \lambda} \left(\sigma^i_x - \frac{\sigma_x(\lambda)}{\sigma(\lambda)} \sigma^i\right).
\end{aligned}
$$
This implies that \eqref{bkm3} is exactly \eqref{bkm_b}, written in coordinates $\sigma^i$, as required.
\end{proof}

We choose an arbitrary polynomial $c(\lambda)$ and take  a solution $w(t,x)=\bigl(w_1(t,x),\dots, w_N(t,x)\bigr)$ of \eqref{eq:BKMreduced}.  Then by  Corollary \ref{cor:1},  
$$
w(t,x,\mu) = \mu^N + w_1(t,x)\mu^{N-1} + \dots + w_N(t,x) = \det \bigl(\mu \,\Id - M(w(t,x)\bigr)
$$ 
satisfies the universal solitonic equation \eqref{soliton} and consequently, by Proposition \ref{p3},  ${\rho}(t,x,\mu)={\rho}\bigl(\mu, w(t,x)\bigr)$ satisfies \eqref{eq1mod2}.

Thus,   we see that $w(\mu)=w(t,x,\mu)$ and ${\rho}(\mu)={\rho}(t, x, \mu)$ satisfy the system
\begin{equation}
\label{eq:wrho}
\begin{aligned}
\partial_t {\rho}(\mu) &= \frac{1}{\mu - \lambda} \Bigl(  m(\mu) w_{xxx}(\lambda) + 2 {\rho}(\mu) w_x(\lambda) + {\rho}_x(\mu)w(\lambda) 
         \Bigr), \\
c(\mu) &= m(\mu) \Big(  w_{xx}(\mu)  w(\mu) - \tfrac{1}{2}  w_x^2(\mu) \Big) + {\rho}(\mu)  w^2(\mu),       
\end{aligned}
\end{equation}
for all $\mu\in\R$ (except $\mu=\lambda$ in the first equation).

Now, it remains to compare \eqref{eq:wrho} with relations \eqref{bkm1_b}  (which are equivalent to \eqref{eq:BKM}) 
and notice that the substitution  
$$
{q} = \frac{w(\lambda)}{\sqrt{c(\lambda)}} \quad \mbox{and}\quad  \sigma(\mu) =   {\rho}(\mu) 
$$ 
transforms systems into each other with the appropriate rescaling of time $t \mapsto \sqrt{c(\lambda)} \,t$, completing the proof of the first part of Theorem \ref{t1}.

The second part (i.e. the case $\lambda = \infty$) can be proved in a similar way.
The system  in question \eqref{eq:BKM3} has the following form
\begin{subequations}
\begin{align}
u_t &= q_{xxx} \zeta + (L + q\,\Id) u_x, \label{eq:BKM3_1}\\
\frac{1}{2}\tr L &= q +  \frac{m_n}{2} q_{xx} \label{eq:BKM3_2}
\end{align}
\end{subequations}
where $m(\mu)=m_0 + m_1\mu + \dots + m_n \mu^n$. 
We start with analysing the evolution of $\sigma(\mu)=\det (\mu\,\Id - L)$:
$$
\begin{aligned}
\sigma_t (\mu) = & \langle \ddd \sigma (\mu), u_t   \rangle = 
\langle \ddd\sigma (\mu), q_{xxx} \zeta + (L + q\,\Id) u_x  \rangle =\\
& q_{xxx} \mathcal L_\zeta \sigma (\mu) +    \langle L^*\ddd \sigma (\mu), u_x  \rangle +
q \sigma_x (\mu).
\end{aligned}
$$
Next we use  identities \eqref{eq:Lzeta},  the relation $L^* \ddd \sigma(\mu) = \sigma(\mu) \ddd \tr L  + \mu\, \ddd\sigma(\mu)$ (see \cite[Proposition 2.2, formula (6)]{nij1}) and constraint \eqref{eq:BKM3_2} to get
$$
\begin{aligned}
\sigma_t (\mu) & = q_{xxx} \bigl(m(\mu) - m_n \sigma(\mu)\bigr) +  
\sigma(\mu) (\tr L)_x  + (\mu + q)\sigma_x(\mu) = \\
& = m(\mu) q_{xxx}  +  \bigl( - m_n q_{xxx} + (\tr L)_x \bigr) \sigma(\mu)   
 + (\mu + q)\sigma_x(\mu) = \\
&  = m(\mu) q_{xxx}  +  2 q_x \sigma(\mu)   
 + (\mu + q)\sigma_x(\mu).
\end{aligned}
$$
Thus, we conclude that \eqref{eq:BKM3} is equivalent to the system of two equations  (cf. Proposition \ref{ps_b})
\begin{equation}
\label{eq:BKM3equiv}
\begin{aligned}
\sigma_t (\mu) &= m(\mu) q_{xxx}  +  2 q_x \sigma(\mu)   
 + (\mu + q)\sigma_x(\mu) \\
\frac{1}{2}\tr L &= q +  \frac{m_n}{2} q_{xx}
\end{aligned}
\end{equation}

We now assume that $w(t,x)$ satisfies \eqref{eq:BKMreduced} with $\lambda = \infty$. Then by Propositions
\ref{prop:infty1} and \ref{prop:infty2},  the function $\rho(\mu)=\rho(t,x,\mu)$ satisfies \eqref{eq:rhoinftymod}. Hence, for  $w$ and $\rho$ we have 
\begin{equation}
\label{eq:some}
\begin{aligned}
        \partial_t {\rho}(\mu) &=   m(\mu) (w_1)_{xxx} +  2  {\rho}(\mu) (w_1)_x +  {\rho}_x(\mu)
        \left( \mu + w_1\right),\\
  c(\mu) & =   m(\mu) \bigl(w_{xx}(\mu) w(\mu) - \tfrac{1}{2} w_x^2(\mu)\bigr)  +   \rho(w,\mu) w^2(\mu)  \quad
  \mbox{for all $\mu\in\R$.}
      \end{aligned}   
      \end{equation}

 Recall that the left hand side and right hand side of the latter equation are both monic polynomials of degree $2N+n$, i.e., this equation is of the form 
 $$
 \mu^{2N+n} + \mu^{2N+n-1} c_1 + \dots =   \mu^{2N+n} + \mu^{2N+n-1} P_1 + \dots
 $$
 where $P_1 = m_n (w_1)_{xx}  +  \rho_1 + 2w_1$,  so that   \eqref{eq:some} implies  
 \begin{equation}
\label{eq:some2}
\begin{aligned}
        \partial_t {\rho}(\mu) &=   m(\mu) (w_1)_{xxx} +  2  {\rho}(\mu) (w_1)_x +  {\rho}_x(\mu),\\
 - \frac {1}{2}\rho_1 & =  \left(w_1- \frac{ c_1}{2}\right) + \frac{ m_n}{2} (w_1)_{xx}.
      \end{aligned}   
      \end{equation}     

Now comparing \eqref{eq:BKM3equiv} and \eqref{eq:some2} (and using the fact that $\sigma_1 = -\tr L$), we see that the substitution
$$
q(t,x) = w_1(t,x) - \tfrac{ c_1}{2} \quad \mbox{and} \quad \sigma\bigl(\mu, u(t,x)\bigr) = \rho \bigl(\mu, w(t, x + \tfrac{ c_1}{2} t) \bigr)   
$$
transforms systems \eqref{eq:BKM3equiv} and \eqref{eq:some2}  to each other, completing the proof of the second part of Theorem \ref{t1}.

\section{Proof of Theorem \ref{t2}}

We consider the base equation  \eqref{eq:1a} as a collection of {\it algebraic} relations for $w, w_x, w_{xx}$ parametrised by $\mu\in \R$.  
Since by construction \eqref{eq:1a} can be written in the form
\begin{equation}
\label{eq:1amod}
\Bigl(w_{xx}(\mu) w(\mu) - \tfrac{1}{2} w_x^2(\mu)\Bigr) + 
\frac{{\rho}(\mu,w) w^2(\mu) - c(\mu)}{m(\mu)}= 0, \quad \mbox{for all $\mu \in \R$}. 
\end{equation}
where the left hand side is a polynomial $P(\mu)$ in $\mu$ of degree $2N-1$,
we can replace this infinite system of relations by $2N$ equations stating that each of $2N$ coefficients of $P(\mu)$ vanishes.  Equivalently,  we may choose $N$ arbitrary distinct values $\mu_1,\dots,\mu_N$ of the parameter $\mu$ and replace the base equation by the following $2N$ relations:
\begin{equation}
\label{eq:c5}
P(\mu_i)=0, \  P'_\mu(\mu_i)=0, \qquad  i=1,\dots, N.
\end{equation}
We will do it for $\mu_1,\dots,\mu_N$ being the roots of $w(\mu)$ (i.e., the eigenvalues of $M$).   

\begin{Proposition}
The base equation \eqref{eq:1a}
is equivalent to the system of first order differential equations
\begin{equation}
\label{eq:c2}
-  \frac{1}{2} w_x^2(y_i)   -  \frac{c(y_i)}{m(y_i)} = 0, \qquad i=1,\dots, N,
\end{equation} 
where $y_i$ are the roots of $w(\mu)=\det(\mu\,\Id - M)$ (i.e., the eigenvalues of $M$).
\end{Proposition}

\begin{proof}
The first half of \eqref{eq:c5}  is obtained by substitution $\mu = y_i$ and $w(y_i)=0$ into the base equation in the form \eqref{eq:1amod} and, hence,  coincides with \eqref{eq:c2}.

The second half  of \eqref{eq:c5} is obtained by differentiating \eqref{eq:1amod} w.r.t. $\mu$ and then substituting $\mu = y_i$. After some simplifications and, in particular, using the fact that the term containing $w^2(\mu)$ disappears, we get ( $'$ denotes the derivative w.r.t. $\mu$)
 \begin{equation}
\label{eq:c6}
 w_{xx} (y_i) w'(y_i)   -   w'_x(y_i) w_x(y_i)   -  \left(\frac{c(\mu)}{m(\mu)}\right)'_{|_{\mu = y_i}} = 0, \quad i= 1,\dots, N.
\end{equation}
 
As explained, \eqref{eq:c2} and \eqref{eq:c6} together are equivalent to the base equation \eqref{eq:1a}.  We now claim that \eqref{eq:c6}, as a differential equation,  follows from \eqref{eq:c2} and for this reason can be ignored (which would complete the proof). Indeed, let us differentiate \eqref{eq:c2} by $x$:
$$
- w_{xx} (y_i) w_x(y_i) - w'_x(y_i) w_x(y_i) \tfrac{\dd y_i}{\dd x} - \left(\frac{c(\mu)}{m(\mu)}\right)'_{|_{\mu = y_i}} \, \tfrac{\dd y_i}{\dd x}  = 0
$$
and then divide by $\tfrac{\dd y_i}{\dd x}$ to get 
\begin{equation}
\label{eq:c7}
- w_{xx} (y_i) w_x(y_i) \left(\, \tfrac{\dd y_i}{\dd x} \right)^{-1}- w'_x(y_i) w_x(y_i) 
 - \left(\frac{c(\lambda)}{m(\lambda)}\right)'_{|_{\lambda = y_i}}  = 0.
\end{equation}
It remains to notice that if we treat $w(\lambda)$ as a function $w(\lambda, x)$ of two variables $\lambda$ and $x$, then the identity
$w(y_i , x)\equiv 0$ gives  $w' (y_i) \tfrac{\dd y_i}{\dd x} + w_x(y_i) = 0$  (implicit function theorem).   The latter shows that \eqref{eq:c7} coincides with \eqref{eq:c6}, as required.  This completes the proof of Proposition. \end{proof}

Next, we need to show that \eqref{eq:c2} is equivalent to \eqref{eq:c3}. Recall that the latter has the form
\begin{equation}
\label{eq:c3bis}
\tfrac{1}{2}g_0 (M_\mu \dot w, \dot w) - V_\mu(w)=0,   \quad \mu \in\R,
\end{equation}
where $V_\mu=U_0\mu^{N-1} + U_1\mu^{N-2} + \dots + U_{N-1}$ and $U_i$ are defined from the matrix relation
\begin{equation}
\label{eq:c4}
f(M) = U_0M^{N-1} + U_1 M^{N-2} + \dots + U_{N-1}\Id, \quad    \mbox{with $f(t) = \dfrac{c(t)}{m(t)}$}.
\end{equation}

The left hand side of \eqref{eq:c3bis} is a polynomial in $\mu$ of degree $N-1$ so that this infinite system of algebraic relations can be replaced by $N$ relations obtained by substitution $\mu=y_i$, $i=1,\dots, N$, where $y_i$ are the roots of  $w(\mu)$, i.e., eigenvalues of $M$.  Also,  the verification can be done in any coordinate system so that we are allowed to use $y_1,\dots,y_N$ as local coordinates for our computations. In this coordinate system,  the operator $M$ takes the canonical diagonal form  $M = \operatorname{diag}(y_1,\dots, y_N)$ and
therefore $M_\mu = \operatorname{diag} \Bigl(- \prod_{i\ne 1} (\mu-  y_i),-  \prod_{i\ne 2} (\mu- y_i), \dots \Bigr)$. The metric $g_0$ (given in the coordinates $w_1,\dots,w_N$ by \eqref{eq:g0})  takes now the form
$$
g_0 = \sum_{\alpha} \prod_{i\ne\alpha} (y_\alpha - y_i) \ \dd y_\alpha^2.
$$
Hence, for $\mu =  y_\alpha$, the first term of \eqref{eq:c3bis}  becomes
$$
\tfrac{1}{2} g_0(M_{y_i} \dot y, \dot y) = - \tfrac{1}{2}\prod _{i\ne\alpha} (y_\alpha - y_i) \prod_{i\ne \alpha} (y_\alpha - y_i) {\dot y_\alpha}^2 =- \tfrac{1}{2} \Bigl(\dot y_\alpha \prod _{i\ne\alpha} (y_\alpha - y_i)\Bigr)^2.
$$
Next, from \eqref{eq:c4} we have
\begin{equation}
\label{eq:Vyalpha}
V_{y_\alpha}(y) = U_0(y) \, y_\alpha^{N-1} + U_1(y) \, y_\alpha^{N-2} + \dots + U_{N-1}(y)= f(y_\alpha) = \frac{c(y_\alpha)}{m(y_\alpha)}.
\end{equation}
Thus,  for $\mu = y_\alpha$, \eqref{eq:c3bis} takes the form:
\begin{equation}
\label{eq:c8}
\frac{1}{2} \Bigl(\dot y_\alpha\prod _{i\ne\alpha} (y_\alpha - y_i) \Bigr)^2 + \frac{c(y_\alpha)}{m(y_\alpha)} = 0. 
\end{equation}

Finally,  we use the fact that $w(\mu) = \prod_{i=1}^N (\mu - y^i)$.  Hence
$
w_x(\mu) = \dot w(\mu) = -\sum_{s}   \dot y^s \prod_{i\ne s} \bigl(\mu -  y^i\bigr)
$
so that for $\mu = y^\alpha$ we get
$$
w_x(y_\alpha) = -  \dot y_\alpha \prod_{i\ne\alpha} \bigl(y_\alpha - y_i \bigr).
$$
Then \eqref{eq:c8} can be written as
$$
\frac{1}{2}   w_x^2(y_\alpha)  + \frac{c(y_\alpha)}{m(y_\alpha)} = 0, 
$$
which coincides with \eqref{eq:c2},  as stated.  This completes the proof of part (i) in Theorem \ref{t2}.

Statements (ii) and (iii) are, in fact, general properties of integrable geodesic flows (with potential) admitting first integrals quadratic in velocities.  In the Hamiltonian setting, this property can be formulated as follows.

\begin{Proposition}
Let $H(w,p) = \frac{1}{2} g^{-1} (p, p) + U(w)$ be the Hamiltonian of the geodesic flow of a (pseudo-)Riemannian metric $g$ with a potential $U$ in the canonical coordinates $(w,p)$ on $T^*\R^N$ and $F=\frac{1}{2} g^{-1} (A^* p, p) + V(w)$ be a first integral of it. Here $A$ is a $g$-selfadjoint operator field, which can be understood as a Killing $(1,1)$-tensor of $g$.

Then the set of the corresponding geodesics $\{ w(x)\}$ is invariant under the quasilinear system 
\begin{equation}
\label{eq:quasi3}
w_t = A w_x.
\end{equation}
Moreover,  if 
$$
\bigl(  w(x,t), p(x,t)   \bigr) = \Phi^x_H  \circ \Phi^t_F (\widehat w, \widehat p), \quad (\widehat w, \widehat p) \in T^*\R^N
$$
is an orbit of the Hamiltonian $\R^2$-action generated by the commuting functions $H$ and $F$, then $w(x,t)$ is a solution of \eqref{eq:quasi3}.
\end{Proposition}

\begin{Remark} {\rm This relationship between integrable geodesic flows and quasilinear systems is well-known and was used in many papers, see  Ferapontov \cite{fer, fer1},  Magri \cite{m}, Magri and Lorenzoni \cite{ml},  Blaszak and Wen-Xiu Ma \cite{Bl2},  and also our recent paper \cite{nijapp5}.   For the sake of completeness, we remind the proof here.} 
\end{Remark}

\begin{proof}
It is sufficient to compare the Hamiltonian equations related to the Hamiltonians $H$ and $F$ 
(we use $x$ and $t$ for the time-variables related to $H$ and $F$ respectively):
$$
\left\{\begin{aligned}
\frac{\ddd w^i}{\ddd x} &= \frac{\partial H}{\partial p_i}= \sum_\beta g^{\beta i}p_\beta\\
\frac{\ddd p_i}{\ddd x} &= \frac{\partial H}{\partial w^i}  
\end{aligned}\right. \quad \mbox{and} \quad
\left\{\begin{aligned}
\frac{\ddd w^i}{\ddd t} &= \frac{\partial F}{\partial p_i}= \sum_{\alpha,\beta} g^{i\alpha} A_\alpha^\beta p_\beta=
 \sum_{\alpha,\beta} A_\alpha^i g^{\beta\alpha}  p_\beta\\
\frac{\ddd p_i}{\ddd t} &= \frac{\partial F}{\partial w^i}  
\end{aligned}\right.
$$
The first equations of these two systems imply $\frac{\ddd w^i}{\ddd t} = A^i_\alpha \frac{\ddd w^\alpha}{\ddd x}$ or, equivalently, 
$w_t = A w_x$ for every common solution $\bigl(w(t,x),p(t,x)\bigr)$ of these two commuting Hamiltonian flows which can be equivalently understood as an orbit of the Hamiltonian $\R^2$-action generated by $H$ and $F$.  

This construction basically shows that the evolution of each geodesic $w(x)$ under the Hamiltonian flow generated by the first integral $F$ coincides with its evolution under the quasilinear system \eqref{eq:quasi3}.  
\end{proof}

In the settings of Theorem \ref{t2}, we have a similar situation.  The only difference is that we need to consider not all $g$-geodesic but only those which are located on a certain common level surface $\mathcal X$ of  commuting integrals $F_\mu$.  This does not affect the conclusion, because $\mathcal X$ is invariant with respect to the both Hamiltonian flows $\Phi^x_H$ and $\Phi^t_{F_\lambda}$.  This completes the proof of Theorem \ref{t2}.

 \section{Conclusion and outlook}

The paper presents  a new method for constructing solutions to the series of integrable multi-component PDE systems introduced in  \cite{nijapp4}. This series 
contains as particular examples (with appropriately chosen parameters)
and generalises many famous integrable systems including KdV, coupled
KdV \cite{fordy}, Harry Dym, coupled Harry Dym \cite{fordy}, Camassa-Holm, multicomponent
Camassa-Holm  \cite{ih}, Dullin-Gottwald-Holm \cite{gdh} and Kaup-Boussinesq
systems, so the method can be applied to all these systems.

Finite-dimensional reductions for KdV were understood in a series of classical works by  Gelfand--Dikij, Kruskal--Zabudsky,  Bogoyavlenskij--Novikov--Dubrovin, Albers,  Krichever,  Veselov, Moser, Van Moerbeke, McKean, Trubowitz, Faddeev, Its, V.\,B.\,Matveev  and others, see e.g. the recent historical survey  \cite{VBM}.  They considered, see e.g. \cite{novikov74},  a family of functions  $u(x)$ satisfying the condition
\begin{equation} \label{eq:lam}
B_{N}[u] = \lambda_{N-1} B_{N-1}[u] + \cdots \lambda_0 B_0[u], 
\end{equation}
where $B_i[\cdot]$ is the $i$-th higher symmetry of the KdV equation. The case $i=1$ corresponds to the KdV flow itself so that $B_1[u]=\tfrac{1}{2}u_{xxx} + \tfrac{3}{2} u u_x$.   In other words,  these are those functions  for which the $N$-th symmetry of KdV is a linear combination of lower order symmetries. The family of such functions is finite-dimensional and  is invariant with respect to the evolutionary flows generated by the KdV equation and all of its higher symmetries.   The restriction of these flows onto this family is a finite-dimensional integrable system \cite{BN}. 

This method was applied to many other integrable PDEs including special cases of BKM equations, see again  \cite{VBM}. However,  for  BKM  systems  with  $n>1$ and $\deg m(\mu)\ge 2$,  this classical approach encounters serious technical difficulties. To overcome them, we do not restrict ourselves to  stationary solutions of higher symmetries, but   directly construct a finite-dimensional system and the mapping $\mathcal R$ which sends its  solutions to those of the BKM system. For KdV and most other integrable systems listed above,  the output, i.e., the finite-dimensional system and its embedding,   coincides with the classical `$N$-stationary' reduction. Our approach, however, has some special features and advantages:

\begin{itemize}

\item One of the key points is that BKM is a family of evolutionary flows parametrised by  $\lambda \in\R \cup\{\infty\}$. These flows pairwise commute and therefore can be understood as symmetries of each other, so that instead of the classical series of symmetries $B_i$, $i\in \mathbb N$, we deal with symmetries depending on a continuous parameter\footnote{We have to pay for this change by introducing an additional differential constraint  (see \eqref{eq:BKM}), but as a reward, this  gives us a possibility to include into our construction many interesting examples including Camassa-Holm type systems.}. After this, we are looking for solutions depending polynomially in $\lambda$, which naturally leads us to a desired reduction.

\item  Our series of integrable PDEs was obtained in the framework of {\it Nijenhuis Geometry} research programme initiated in \cite{nij1}.  It is constructed from a Nijenhuis operator $L=L(u)$ defined on $\R^n(u^1,\dotsm u^n)$ and an arbitrary polynomial $m(\lambda)$ of degree $n$,  where $n$ is the number of components $u^1,\dots, u^n$. The operator $L$ satisfies certain non-degeneracy conditions and for this reason is unique up to coordinate transformations, whereas the polynomial $m(\lambda)$ parametrises these series and essentially affects their properties. 

\item Remarkably, the reduction procedure also goes through Nijenhuis geometry, as its first step reduces the initial system to an integrable systems of hydrodynamic type studied, in the context  Nijenhuis geometry, in \cite{nijapp4,nijapp5}. 
Similar to the initial PDE system, the reduced system is also based on a Nijenhuis operator $M$ defined on $\R^N$.  All the other ingredients, such as the metric $g_0$ and potential $U_0$ which determine the dynamics, are naturally constructed from $M$. This fact provides another interesting link between Nijenhuis geometry and integrable systems. 

\item As explained in  \S \ref{sec:1.4}, the reduced system can be integrated by separation of variables.  However, the separating variables are not globally defined now on $\R^N$ and  have singularities at some points. In terms of the operator $M$, such points correspond to collision of its eigenvalues.  Inclusion of such singular points plays  an important role. Indeed,  many important solutions of KdV and other previously studied systems, for  examples,  soliton-like solutions must pass through them.    Recall that the study of singular points is an important direction in the Nijenhuis geometry research program suggested in \cite{nij1}.

\item Our method is quite general in the sense it  can be uniformly applied to all  BKM systems. Moreover, the finite-dimensional integrable  system obtained by reduction, is  the same for all types of BKM systems and  belongs to the type of well-studied integrable Hamiltonian systems, where commuting functions are sums of  kinetic and potential terms and coordinate separation of variables works almost everywhere on the phase space.

\item Along with classical works on finite-dimensional reductions of integrable PDEs, we would like to highlight a series of papers by  Blaszak, Marciniak and Szablikowski \cite{BM2006, BM2008, BS2023, MB2010, SBM24}. The authors of these papers took a different route:  the primary object of their investigation is   a finite-dimensional  Benenti type system from which they  come to  certain  integrable PDE systems which are special cases of BKM IV systems.    We note however that their integrable system is visually similar, but less general than the system in our paper; in particular, in the KdV case they considered the equation \eqref{eq:lam} with all $\lambda_i=0$.  Namely, in our notation, their rational function $c/m$ is always a polynomial. Moreover, the coefficients $c_1,c_2,\dots, c_{N+n-1}$ of the polynomial $c(\mu)$ are assumed to be zero.  

\end{itemize}

BKM systems generalise many  well-known and well-studied  equations in mathematical physics. Our  general goal was to understand whether the methods developed for these equations can be extended to general BKM systems. The results of the present paper show that the finite-dimensional reduction method does work. Moreover, our construction uniformly applies to  all BKM systems and, in particular, allows one to avoid case by case study.  The next goal would be to understand whether other classical methods  work for general BKM systems too. In particular, we expect that Lax pair and inverse scattering method can be generalised for BKM systems. We also hope 
that the methods for constructing exact solutions using special functions can be applied here, and point out possible difficulties in \S \ref{sec:1.4}.   We may need a little help on these issues, and invite colleagues, especially those with background in the classical theory of integrable PDEs and algebraic  geometry,  to join the investigation.

\printbibliography

\end{document}